\documentclass[12pt, draftclsnofoot, onecolumn]{IEEEtran}
\ifCLASSINFOpdf
\else
\fi
\usepackage{graphicx}
\usepackage[numbers,sort&compress]{natbib}
\usepackage{amsthm}

\usepackage{amsmath}
\usepackage{amssymb}
\usepackage{bbm}
\usepackage{graphicx}
\usepackage[space]{grffile}
\usepackage{subfig}
\usepackage[font=scriptsize]{caption}

\theoremstyle{remark}
\newtheorem{thm}{$\;\;\;$Theorem}
\newtheorem{lem}{$\;\;\;$Lemma}
\newtheorem{prop}{$\;\;\;$Proposition}
\newtheorem{cor}{$\;\;\;$Corollary}
\newtheorem{defn}{$\;\;\;$Definition}

\usepackage{xcolor}

\begin{document}
%
\title{Reversing the Curse of Densification in mmWave Networks Through Spatial Multiplexing}
%
%
%

\author{Shuqiao~Jia and
        Behnaam~Aazhang~\IEEEmembership{}}

%
%

\markboth{draft}%
{Shell \MakeLowercase{\textit{et al.}}: Bare Demo of IEEEtran.cls for Journals}
%



\maketitle

\vspace*{-20pt}
\begin{abstract}
The gold standard of a wireless network is that the throughput increases linearly with the density of access points (APs). However, such a linear throughput gain is suspended in the 5G mmWave network mainly due to the short communication distances in mmWave bands and the dense deployments of mmWave APs. As being operated in the interference-limited regime, the aggregate interference resulted from the increasing mmWave APs will gradually become the network performance bottleneck, which leads to the saturation of the throughput.
In this paper, we propose to overcome the densification plateau of a mmWave network by employing spatial multiplexing at APs.
To study the effect of spatial multiplexing on mmWave networks, we first derive the coverage probability as a function of spatial multiplexing gain. 
The fixed-rate coding scheme is then used to provide the network throughput. 
We also introduce the concept of densification gain to capture the improvement in network throughput achieved through the AP densification.
Our results indicate that without the spatial multiplexing at APs, the throughput of mmWave networks will reach the plateau when the density of APs becomes sufficiently large.
By enabling the spatial multiplexing at APs, however, the mmWave network can continuously harvest the throughput gains as the number of APs grows.
By deriving the upper bound for the throughput of mmWave networks, we quantify the potential throughput improvement as the spatial multiplexing gain increases.
However, our numerical results show that such a potential throughput gain cannot be explored by the fixed-rate coding scheme. We then demonstrate the necessity for deploying the multi-rate coding scheme in mmWave networks, especially when the spatial multiplexing gain at mmWave APs is large.

\end{abstract}


%
\IEEEpeerreviewmaketitle

\section{Introduction}
Densifying the deployments of access points (APs) is one of the most significant drivers for increasing the throughput of wireless network \cite{gupta2000capacity,tse2005fundamentals}. As the radio resources can be reused across smaller spatial scales, the network throughput is expected to increase linearly with the AP density \cite{tse2005fundamentals}.
However, the sustainability of such a linear throughput gain with the AP densification has been challenged in the recent studies of 5G mmWave networks \cite{andrews2017modeling,bai2015coverage,jia2016impact,bai2014analysis,saha2018integrated,alammouri2018unified}, mainly due to the inherently high density of mmWave LAN infrastructure \cite{jia2016impact,bai2015coverage,alammouri2018sinr,rappaport2013broadband,rappaport2013millimeter,rappaport2015wideband}. 
Specifically, the mmWave network is most likely to be operated in the interference-limited regime and thus the throughput will reach the plateau when the aggregate interference from the APs becomes overwhelming \cite{bai2015coverage,andrews2017modeling}.
To suppress the interference, the hybrid precoding is performed at the mmWave device, where the hybrid precoding technique is the combination of analog beamforming and spatial multiplexing \cite{andrews2017modeling,bai2015coverage,alkhateeb2015limited,sun2014mimo}. However, it has been proved in \cite{alammouri2020escaping} that the densification plateau cannot be avoided with the finite analog beamforming gain. In this paper, we propose to overcome the throughput densification plateau by deploying the spatial multiplexing at mmWave APs. With that goal, we provide the comprehensive performance analysis for the downlink mmWave networks.

In the context of dense APs and strong interference, we use the stochastic geometry to facilitate the performance analysis for the mmWave network, where the APs and end-users are modeled as two independent Poisson Point processes (PPPs).
The use of stochastic geometry in modeling the communication system has been widely accepted due to its accuracy and tractability in analyzing the signal-to-interference-plus-noise ratio (SINR) distribution for the wireless network \cite{andrews2011tractable,haenggi2012stochastic}.
In \cite{andrews2017modeling, bai2015coverage, di2015stochastic, singh2015tractable}, the stochastic geometry has been extended to study the SINR distribution for mmWave networks by incorporating the distinguishing features of the mmWave system. 
For example, to model the high penetration loss of mmWave signal, the study in \cite{di2015stochastic, bai2015coverage, andrews2017modeling} introduces the mmWave link state, namely the line-of-sight (LOS) and the non-line-of-sight (NLOS) state. 
In \cite{alkhateeb2015limited, bai2015coverage, di2015stochastic, kulkarni2017performance,saha2018integrated}, the mmWave link is modeled with the sectorized beam patterns, which characterizes the analog beamforming gain of the hybrid precoding performed at the mmWave device.
In addition to the analog beamforming gain, we introduce the spatial multiplexing gain to capture the impact of multiple RF chains on the mmWave hybrid precoding \cite{jia2016impact}.

To take account of the blockage effects on the mmWave network, we use the relative density to quantify the number of APs throughout this paper. The relative density is defined as the average number of LOS APs that an end-user can observe, which was first introduced in \cite{bai2015coverage}.
In Section II, we propose to model the mmWave network by incorporating the hybrid precoding architecture in \cite{alkhateeb2015limited}, where both the analog beamforming gain and the spatial multiplexing gain are considered. 
In Section III, we derive the analytical expressions of the coverage probability, the throughput of the fixed-rate coding scheme, and the throughput upper bound for mmWave networks. 
The analysis in Section III is then used to investigate the effect of AP densification on mmWave networks in Section IV. Also, the results of Section III provide the basis for the performance evaluation of the mmWave network throughput in Section V.

Our results provide two key system design insights for the mmWave network. In Section IV, we analyze how the throughput of a mmWave network is affected by the AP density, where we assume that the fixed-rate coding scheme is employed. Moreover, we introduce a key performance indicator, termed densification gain, to capture the throughput scaling law as the AP density increases. 
Our analysis shows that without the spatial multiplexing at the APs, the throughput of a mmWave network will reach the plateau when the AP density grows very large. However, by employing spatial multiplexing, the mmWave network can continuously harvest the throughput gain from the AP densification.

In Section V, we numerically quantify the throughput of different coding schemes for the mmWave network.
The fixed-rate coding scheme with the optimal rate threshold is used to provide the tightest throughput lower bound, while the multi-rate coding scheme is developed for further throughput improvements. We also numerically compare these two schemes to the throughput upper bound as the spatial multiplexing gain increases.
Our results suggest that for the mmWave network, the performance of the fixed-rate coding scheme is less satisfactory with the larger spatial multiplexing gain, where the multi-rate coding scheme is required to explore the potential throughput gain.

\section{Modeling mmWave Network}

We consider the downlink mmWave network, where the end-users and APs are distributed in the area according to independent PPPs $\Phi_{\text{U}}$ and $\Phi_{\text{A}}$ with intensities $\lambda_{\text{U}}$ and $\lambda_{\text{A}}$, respectively. 
Each mmWave AP is assumed to be equipped with $k$ RF chains. The mmWave end-user is assumed to be equipped with one RF chain due to its small form factor. 

 \subsection{Hybrid Precoding}
 The mmWave hybrid precoding consists of the analog beamforming at the antenna arrays and spatial multiplexing in the baseband \cite{andrews2017modeling,bai2015coverage}. 
Assume that the mmWave AP is equipped with $k>1$ RF chains, then the AP can simultaneously transmit independent data streams to multiple end-users.
 In the following, we introduce the spatial multiplexing gain to denote the number of data streams transmitted by the mmWave AP. 
Note that the spatial multiplexing gain captures the impact of multiple RF chains on the hybrid precoding.
 \begin{defn}
 	The spatial multiplexing gain for a mmWave AP is defined as the number of independent data streams transmitted by the AP. 
 	\label{def: spatial multiplexing gain}
 \end{defn}
Note that the spatial multiplexing gain of a mmWave AP is limited by the number of RF chains equipped by the AP. Assume that all the RF chains are activated, then the spatial multiplexing gain at the mmWave AP equals to $k$.
 For the mmWave end-user, the spatial multiplexing gain is always one since the end-user has only one RF chain due to the small form factor. 
 It follows that an AP can simultaneously transmit to $k$ end-users.
 
 For the mmWave device, the analog beamforming can be characterized by the main lobe gain, the main lobe beamwidth and the side lobe gain \cite{andrews2017modeling}. We denote the main lobe gain of the mmWave AP as $G_{\text{A}}$, while the side lobe gain is denoted by $g_{\text{A}}$, $g_{\text{A}} \ll G_{\text{A}}$. Let $\theta_{\text{A}}$ denote the main lobe width of the mmWave AP, where we assume $\theta_{\text{A}} \leq \frac{2\pi}{k}$.
 For the mmWave end-user, we denote $\theta_{\text{U}}$ as the main lobe width, where the analog beamforming gain of the main lobe and the side lobe are denoted by $G_{\text{U}}$ and $g_{\text{U}}$, respectively. 
 
 Note that the hybrid precoding is the combination of the spatial multiplexing and analog beamforming \cite{bai2015coverage,andrews2017modeling}. 
In the implementation of hybrid precoding, the area centered at a mmWave AP is equally partitioned into $k$ spatially orthogonal sectors with $\frac{2\pi}{k}$ angular width, as shown in Fig.\ref{fig:spatial multiplexing gain}. 
We assume that all mmWave APs have full buffers, meaning that there is at least one user at each sector of the AP.
Due to the lack of scattering, the mmWave channel is almost deterministic, which means that the channels of nearby end-users are highly correlated \cite{bai2015coverage,saha2018integrated}.
To minimize the inter-beam interference in such a scenario, we assume that a mmWave AP assigns one analog beam to each sector, where the end-users within the same sector are scheduled via the orthogonal time-division and frequency-division techniques.
We remark that the end-users located in different sectors can reuse the same time-frequency resources at the AP, which has been referred to as the sectorization gain \cite{yang2000sectorization}. 

 We assume that the analog beam of an end-user and its intended AP are aligned. However, the analog beams of interfering APs are considered to be randomly oriented for the end-user. Denote the hybrid precoding gain of the mmWave AP as $\mathbb{G}_\text{A}$. For an unintended mmWave link, the hybrid precoding gain of the AP can then be written as $\mathbb{P}(\mathbb{G}_\text{A} = G_\text{A}) = \frac{k\theta_{\text{A}}}{2\pi}$ and $\mathbb{P}(\mathbb{G}_\text{A} = g_\text{A}) = 1 - \frac{k\theta_{\text{A}}}{2\pi}$, while the hybrid precoding gain of the end-user is denoted by $\mathbb{G}_\text{U}$ with $\mathbb{P}(\mathbb{G}_\text{U} = G_\text{U}) = \frac{\theta_{\text{U}}}{2\pi}$ and $\mathbb{P}(\mathbb{G}_\text{U} = g_\text{U}) = 1 - \frac{\theta_{\text{U}}}{2\pi}$.
 In this paper, we focus on the effect of the spatial multiplexing gain $k$ on the performance of mmWave network. Thus, the analog beamforming gains are assumed to be finite constants.
 
 \begin{figure}
 	\vspace*{-10pt}
 	\centering
 	\subfloat[{\scriptsize AP with $k=3$ RF chains i.e. spatial multiplexing gain $k = 3$.} \label{Figure: AP topology with k = 3}]{%
 		\includegraphics[width=0.48\linewidth]{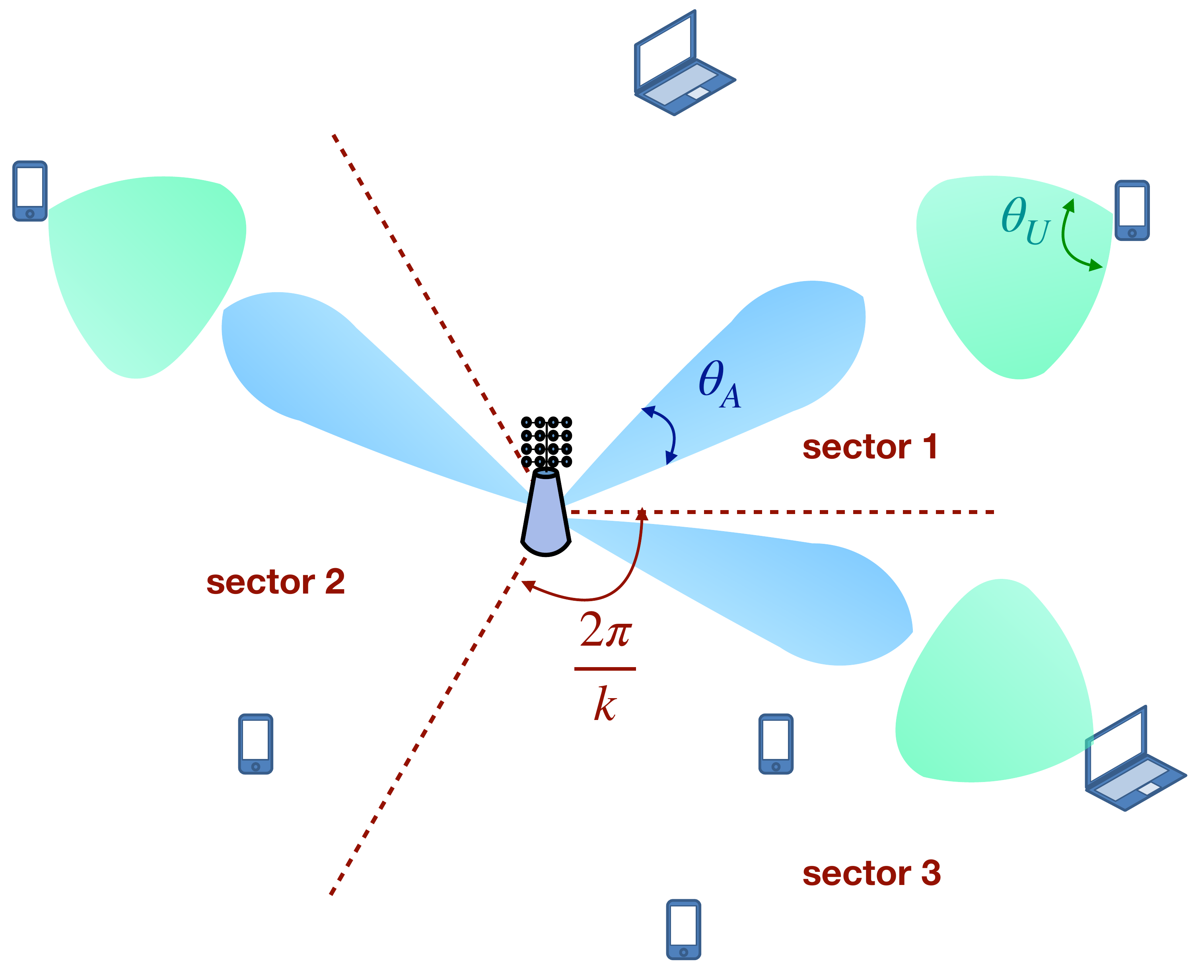}}
 	\hfill
 	\subfloat[{\scriptsize AP with $k=6$ RF chains i.e. spatial multiplexing gain $k = 6$.}\label{Figure: AP topology with k = 6}]{%
 		\includegraphics[width=0.48\linewidth]{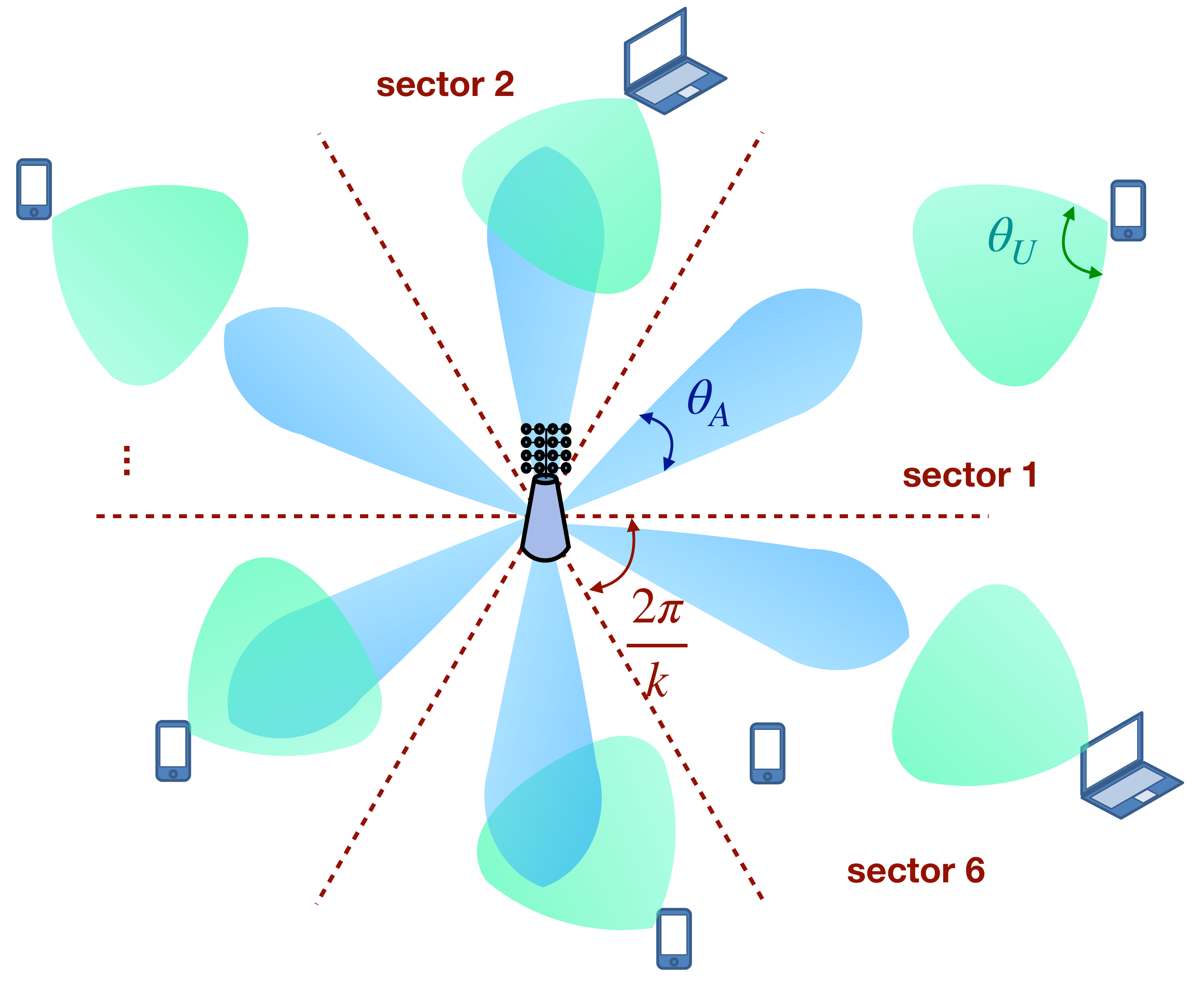}}
 	\caption{\textcolor{black}{Illustration of hybrid precoding at mmWave AP with different number of RF chains}}
 	\label{fig:spatial multiplexing gain} 
 	\vspace*{-20pt}
 \end{figure}
 
\subsection{Relative Density}
The mmWave system experiences a much higher penetration loss than the sub-6GHz system, which leads to dramatically different path loss for the line-of-sight (LOS) and non-line-of-sight (NLOS) mmWave links. 
Denote $\ell(r)$ as the path loss of a mmWave link with length $r$.
To capture the blockage effects in the mmWave network, the path loss model $\ell(r)$ has two states, namely a LOS state and a NLOS state \cite{singh2015tractable,andrews2017modeling}. If the mmWave link is LOS, then the path loss $\ell(r)$ is proportional to $r^{\alpha_{\text{L}}}$ \cite{singh2015tractable}. Such a model can be applied to a NLOS link by replacing the path loss exponent $\alpha_\text{L}$ with $\alpha_{\text{N}}$ \cite{singh2015tractable,andrews2017modeling,rappaport2013millimeter}. 

For the mmWave network with dense deployments of APs, it is shown in \cite{bai2015coverage,andrews2017modeling} that the path loss model can be further simplified to an equivalent LOS-ball of radius $R_\text{B}$ with the inverse
\begin{equation}
\ell(r)^{-1} = 
\left\lbrace 
\begin{split}
& r^{-\alpha_\text{L}}, \; \text{if}\;r\leq R_\text{B}\\
& 0,  \qquad \text{if}\;r> R_\text{B}
\end{split},
\right.
\label{LOS-ball} 
\end{equation}
with the path loss exponent $\alpha_\text{L}\leq 2$.
\textcolor{black}{Moreover, the mmWave link is deterministic if the density of APs is large, which means that the mmWave channel fading is negligible \cite{bai2015coverage}.
By ignoring the channel fading, the coverage of a mmWave AP is equivalent to the intersection of the Voronoi cell and the disk of radius $R_\text{B}$, as shown in Fig.\ref{fig:voronoi_diagram}.
For the mmWave end-user, we then introduce the relative density to characterize the number of mmWave APs within its LOS region.}
\begin{prop}
	In a mmWave network, the relative density $\psi$ is defined as the average number of  LOS APs that an end-user can observe. Given the intensity $\lambda_{\text{A}}$ of $\Phi_{{\text{A}}}$, the relative density can then be given as 
	\begin{eqnarray}
	\psi = \pi \lambda_{\text{A}} R_\text{B}^2,
	\label{eq: relative density}
	\end{eqnarray}
	where $R_\text{B}$ is the radius of LOS-ball in (\ref{LOS-ball}).
		\label{prop: relative density}
\end{prop}
 \vspace{-10pt}
\begin{proof}
	Note that $\Phi_{\text{A}}$ follows a PPP, which is translation-invariant \cite{haenggi2012stochastic}. Thus, the average number of LOS APs is identical for each mmWave end-user. The result then follows the Campbell's Theorem for sums on a PPP, as shown in \cite[Theorem 4.1]{haenggi2012stochastic}.
\end{proof}
It can be observed from (\ref{eq: relative density}) that with a fixed $R_\text{B}$, the relative density $\psi$ is proportional to the intensity of mmWave APs $\lambda_{\text{A}}$, where $R_\text{B}$ is determined by the propagation environment \cite{bai2015coverage}. In our analysis, $R_\text{B}$ is assumed to be a constant. It follows that the increase in the relative density $\psi$ in (\ref{eq: relative density}) is equivalent to the increase in the AP density $\lambda_{\text{A}}$. Thus, we use $\psi$ and $\lambda_{\text{A}}$ alternatively in the sequel.

\begin{figure}[!t]
	\vspace{-0.1cm}
	\centering
	\includegraphics[width=2.5in]{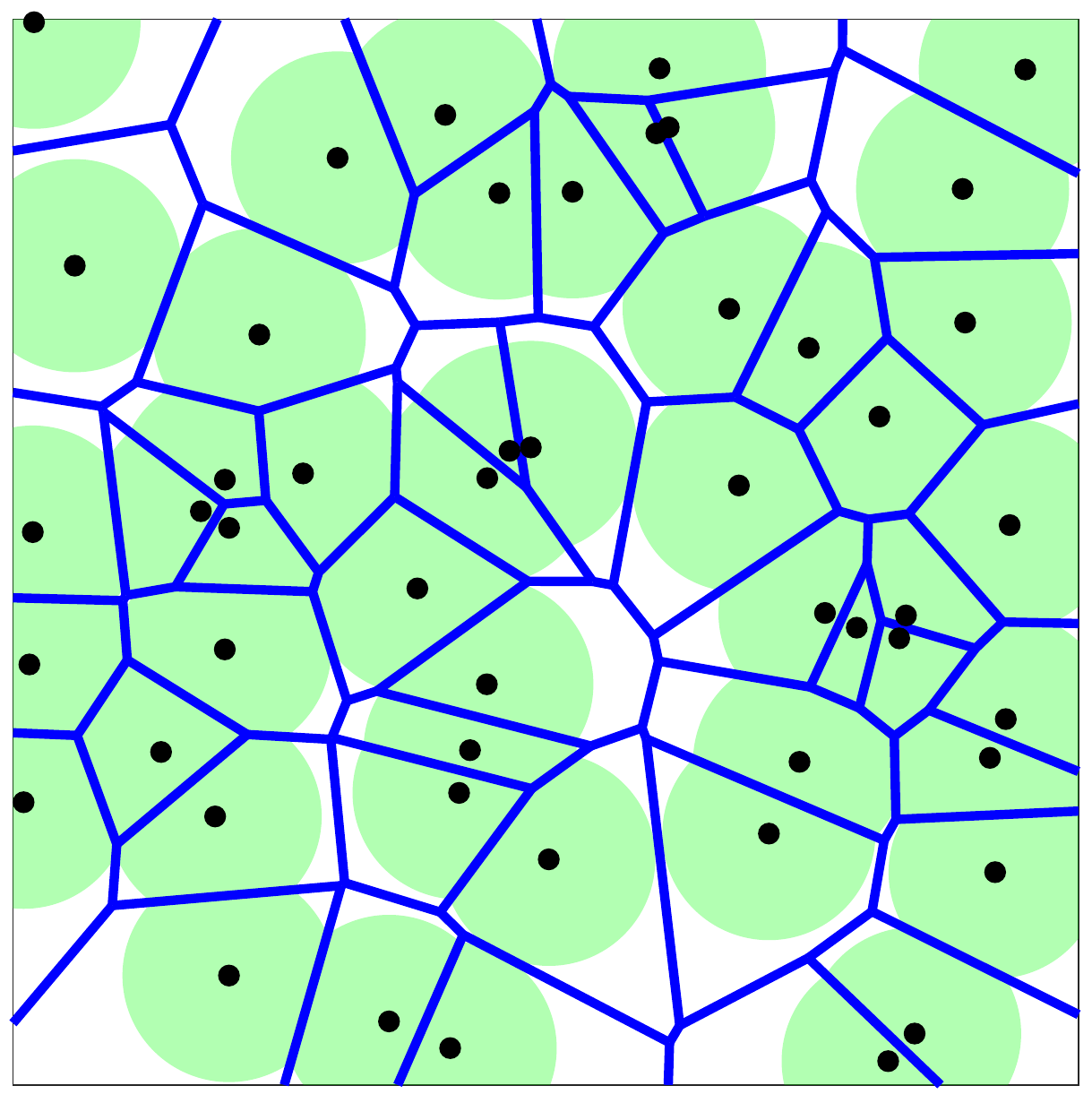}
	\caption{Coverage of mmWave APs. Black dots represent the mmWave APs. For the end-user in the green area, it will be served by the AP located in the same Voronoi cell. The end-user in the white area is out of coverage.}
	\label{fig:voronoi_diagram} 
	\vspace{-20pt}
\end{figure}

\subsection{SINR and Coverage Probability}
In an interference-limited network, the received SINR is considered as one of the most important performance indicators \cite{haenggi2012stochastic,andrews2011tractable}.
Given the SINR threshold $\tau$, the coverage probability of mmWave network with the relative density $\psi$ is then defined as the probability that a randomly chosen end-user can achieve a SINR larger than $\tau$ i.e. 
\begin{equation}
\mathcal{C}_\psi(\tau,k) \triangleq \mathbb{P}\left(\text{SINR}>\tau \left| \mathbb{P}(\mathbb{G}_\text{A} = G_\text{A}) = k\cdot\frac{\theta_{\text{A}}}{2\pi} \right. \right).
\label{def: P_cov}
\end{equation}
Obviously, the coverage probability $\mathcal{C}_\psi(\tau,k)$ increases as the SINR threshold $\tau$ decreases.
Note that the distribution of mmWave APs $\Phi_{\text{A}}$ follows a PPP, which is translation-invariant \cite{haenggi2012stochastic}. Thus the SINR distribution of each end-user is identical. By randomly selecting an end-user from $\Phi_{\text{U}}$, we then shift the origin of our coordinate system to the location of that end-user.

Assume that the end-user at the origin is associated with the mmWave AP located at $x_0$. 
To maximize the signal strength at the origin, the location $x_0$ is selected such that $x_0 = \text{arg}\max\limits_{x\in\Phi_{{\text{A}}}}\ell(|x|)^{-1}$. Such a strategy is widely adopted in the implementation of wireless network, which is known as the minimum path loss association rule \cite{andrews2011tractable,haenggi2012stochastic,bai2015coverage,di2015stochastic,andrews2017modeling,saha2018integrated}.
By following $x_0 = \text{arg}\max\limits_{x\in\Phi_{{\text{A}}}}\ell(|x|)^{-1}$, all the interfering APs are farther than distance $|x_0|$.
Given $x_0$, the received power at the origin resulted by the interfering APs can then be written as
\begin{eqnarray}
\mathcal{I}&=& \sum_{x \in \Phi_{\text{A}} \cap \mathcal{B}(|x_0|,R_\text{B})} \mathbb{G}_\text{A}\mathbb{G}_\text{U}|x|^{-\alpha_{\text{L}}},
\label{eq: interference of LOS}
\end{eqnarray}
where $\mathcal{B}(|x_0|,R_\text{B})$ denotes the ring centered at the origin and of the radius ranging from $|x_0|$ to $R_\text{B}$; $\mathbb{G}_\text{A}, \mathbb{G}_\text{U}$ represent the hybrid precoding gain of the mmWave AP and end-user, respectively. 
Conditioning on $x_0$, the coverage probability $\mathcal{C}_\psi(\tau,k)$ in (\ref{def: P_cov}) can then be written as
\begin{equation}
\mathcal{C}_\psi(\tau,k)
=  \mathbb{P}\left(G_{\text{A}}G_{\text{U}} |x_0|^{-\alpha_{\text{L}}} > {\tau \mathcal{I}} \right).
\label{eq: LOS coverage probability}
\end{equation}
Here, the thermal noise is ignored \cite{bai2015coverage}. Due to $G_\text{A} \gg g_\text{A}$, the signal from the main lobe of $x\in\Phi_{\text{A}}\backslash\{x_0\}$ is considered as the interference by the end-user at the origin, while the signal from side lobes is treated as the background noise.
Next, we derive the coverage probability in (\ref{eq: LOS coverage probability}) based on Alzer's Lemma in \cite{alzer1997some}.
\begin{thm}
	For the mmWave network of relative density $\psi$ and spatial multiplexing gain $k$, the coverage probability $\mathcal{C}_\psi(\tau,k)$ at the SINR threshold $\tau$ is given as
	\begin{eqnarray}
	\mathcal{C}_{\psi}(\tau,k) 
	&=&
	\psi e^{-\psi}{\sum_{j =1}^{\mu}(-1)^{j+1}{\mu\choose j}} 
	\int_{0}^{1} \exp 
	\left\lbrace
	\psi 
	\mathbb{E}
	\left[
	\Lambda_j
	\left(\tau,r,\frac{\mathbb{G}_\text{A} \mathbb{G}_\text{U}}{G_\text{A}G_\text{U}} \right)
	\right]
	\right\rbrace
	\text{d} r,
	\label{eq: simplified coverage probability}
	\end{eqnarray}
	where $\Lambda_j(\tau,r,\cdot)$ is the incomplete gamma function defined in (\ref{eq: taylor parameter}); $\mu$ is the shape parameter of small scale fading channel; $\mathbb{G}_\text{A}, \mathbb{G}_\text{U}$ denote the hybrid precoding gain of mmWave AP and mmWave end-user, respectively; the thermal noise is ignored.
	\label{thm: coverage probability for dense mmWave network}
\end{thm}

\begin{proof}
	Refer to Appendix \ref{proof: coverage probability of LOS ball}.
\end{proof}

To calculate the coverage probability $\mathcal{C}_\psi(\tau,k)$, we choose the shape parameter as $\mu = 10$, which has been validated for modeling the deterministic mmWave channel \cite{bai2015coverage}. Note that the incomplete gamma function $\Lambda_j(\tau,r,\cdot)$ of the coverage probability $\mathcal{C}_\psi(\tau,k)$ in (\ref{eq: simplified coverage probability}) can be fast evaluated by using most numerical tools. In other words, $\mathcal{C}_\psi(\tau,k)$ in Theorem \ref{thm: coverage probability for dense mmWave network} can be efficiently computed with a simple numerical integral over a finite interval. In the following, we continue to show that the coverage probability $\mathcal{C}_{\psi}(\tau,k)$ is a monotonically decreasing function of $k$.

\begin{cor}
	For the mmWave network of relative density $\psi$, the coverage probability $\mathcal{C}_{\psi}(\tau,k)$ at the SINR threshold $\tau$ decreases with the spatial multiplexing gain $k$, $1 \leq k\leq \frac{\theta_{\text{A}}}{2\pi}$. 
	\label{cor: bound of coverage probability}
\end{cor}
\begin{proof}
	Note that the interference $\mathcal{I}$ in (\ref{eq: interference of LOS}) monotonically increases with the hybrid precoding gain $\mathbb{G}_\text{A}$, where $\mathbb{G}_\text{A}$ increases with $k$. The Corollary immediately follows.
\end{proof}

In the prior work \cite{bai2015coverage,di2015stochastic}, the coverage probability of mmWave network is derived without considering the spatial multiplexing at APs, which is equivalent to setting the spatial multiplexing gain as $k=1$. Therefore, the results in \cite{bai2015coverage,di2015stochastic} correspond to the special case of $\mathcal{C}_\psi(\tau,1)$ in (\ref{eq: simplified coverage probability}).

 \subsection{Coding Scheme and Throughput}
 Note that the SINR threshold $\tau$ of the coverage probability $\mathcal{C}_\psi(\tau,k)$ in (\ref{def: P_cov}) is determined by the coding scheme applied to the mmWave network.
 Assume that the mmWave network employs the fix-rate coding scheme at the target rate threshold $\rho$. Then the transmission is considered successful only if the maximum achievable rate of the end-user exceeds $\rho$, i.e., $W \log_2(1+\text{SINR}) > \rho$, where $W$ denotes the bandwidth assigned to the end-user.
 It follows that with the fixed-rate coding scheme of $\rho$, an end-user can achieve a rate
 \begin{equation}
 \mathcal{R} =
 \left\lbrace 
 \begin{split}
 & \rho, \;\text{if}\;\;W \log_2(1+\text{SINR}) > \rho\\
 & 0,  \;\text{if}\;\;W \log_2(1+\text{SINR}) \leq \rho
 \end{split}.
 \right.
 \label{eq: data rate of end-user}
 \end{equation}
Note that the downlink throughput of a mmWave network is defined as the aggregate rate of all end-users.
By implementing the fixed-rate coding scheme with the rate threshold $\rho$, it follows from (\ref{eq: data rate of end-user}) that the throughput can be written as
\begin{eqnarray}
\mathcal{T}(\rho; \psi,k) 
&\triangleq& \lambda_{\text{U}}\mathcal{R} = \lambda_{\text{U}}\rho\mathbb{P}\left(W \log_2(1+\text{SINR})>\rho\right) \nonumber\\
&=&  \lambda_{\text{U}} \rho \mathbb{P}\left(\text{SINR}>2^{\rho/W}-1\right) 
= \lambda_{\text{U}} \rho \mathcal{C}_\psi(2^{\rho/W}-1, k).
\label{eq: throughput of fixed-rate coding scheme}
\end{eqnarray}
Obviously, $\lambda_{\text{U}} \rho$ is the upper bound for the throughput $\mathcal{T}(\rho; \psi,k)$ in (\ref{eq: throughput of fixed-rate coding scheme}), which is achieved only if the interference at  end-users is completely canceled. However, with the finite analog beamforming gain and dense APs, the interference of the mmWave network is strong and thus significantly affects the throughput, which implies $\mathcal{T}(\rho; \psi,k) \ll \lambda_{\text{U}} \rho$.

 A natural extension of the fixed-rate coding scheme is the multi-rate coding scheme. Instead of choosing one rate threshold $\rho$, the multi-rate coding scheme selects a finite set of rate thresholds $\{\rho_i\}_{i=1}^M,M>1$. By applying the multi-rate coding scheme to the mmWave network, the rate of an end-user can then be written as
 \begin{equation}
\hat{\mathcal{R} }=
\left\lbrace 
\begin{split}
& \rho_M, \;\text{if}\;\;W \log_2(1+\text{SINR}) >\rho_M\\
& \rho_i, \;\text{if}\;\; \rho_i < W \log_2(1+\text{SINR}) \leq \rho_{i+1}, \text{for} \;1 \leq i \leq M-1\\
& 0,  \;\text{if}\;\;W \log_2(1+\text{SINR}) \leq \rho_1
\end{split}.
\right.
\label{eq: data rate of end-user - multi-rate coding scheme}
\end{equation}
Thus, the throughput of the mmWave network with the multi-rate coding scheme is given as
\begin{eqnarray}
	&&\mathcal{T}\left( \{\rho_i\}_{i=1}^{M} \,;\, \psi,k \right) 
	\triangleq \lambda_{\text{U}}\hat{\mathcal{R}} \nonumber\\
	&&= \lambda_{\text{U}} \sum_{i=1}^{M-1} \rho_i 
	\left[\mathcal{C}_\psi(2^{\rho_i/W}-1, k) - \mathcal{C}_\psi(2^{\rho_{i+1}/W}-1, k) \right] + \lambda_{\text{U}}\rho_M\mathcal{C}_\psi(2^{\rho_M/W}-1, k),
	\label{eq: throughput of multi-rate coding scheme}
\end{eqnarray}
which is bounded by
\begin{equation}
\mathcal{T}(\rho_1; \psi, k) \leq \mathcal{T}\left( \{\rho_i\}_{i=1}^{M} \,;\, \psi,k \right)  < \lambda_{\text{U}} \rho_M \mathcal{C}_\psi(2^{\rho_1/W}-1, k).
\label{eq: bound on throughput of multi-rate coding scheme}
\end{equation}
Note that the throughput $\mathcal{T}\left( \{\rho_i\}_{i=1}^{M} \,;\, \psi,k \right) $ can be obtained by evaluating the coverage probability in (\ref{eq: simplified coverage probability}) at each SINR threshold $\tau_i = 2^{\rho_i/W}-1, i \in [1,M]$. Hence, there is no difference in analyzing the throughput for the fixed-rate coding scheme and multi-rate coding scheme. Without loss of generality, we will focus on the analysis of fixed-rate coding scheme in the sequel.

For a mmWave network, the upper bound of the throughput is achieved if each end-user can instantaneously reach the maximum achievable rate $W\log_2(1+\text{SINR})$. Therefore, the throughput upper bound for the mmWave network of $\psi$ and $k$ is defined as 
\begin{equation}
\overline{\mathcal{T}}\left(\psi,k \right) = \lambda_{\text{U}}\mathbb{E}\left[ W\log_2(1+\text{SINR}) \right].
\label{eq: throughput upper bound definition}
\end{equation}
To achieve the upper bound in (\ref{eq: throughput upper bound definition}), the mmWave network needs to work with any arbitrarily small SINR, which may not be feasible in practice.
However, we can evaluate the performance of a coding scheme by comparing the corresponding throughput with the upper bound $\overline{\mathcal{T}}\left(\psi,k \right)$.
 
The rest of the paper focuses on the impact of relative density $\psi$ and spatial multiplexing gain $k$ on the throughput of mmWave networks. Other system parameters are assumed to be constants as follows.	
The analog beamforming gain $G_{\text{A}} = 20\,\text{dB}, G_{\text{U}} = g_{\text{A}} = 0\,\text{dB}, g_{\text{U}} = -10\,\text{dB}, \theta_{\text{A}} = 30^{\circ}$ and $\theta_{\text{U}} = 90^{\circ}$ are used in all the numerical results \cite{bai2015coverage,rappaport2013millimeter}. The radius of the LOS-ball in (\ref{LOS-ball}) is assumed to be $R_\text{B} = 200$m and $\alpha_\text{L} = 2$ is used in the numerical results \cite{bai2015coverage,andrews2017modeling}. The total available bandwidth of $2$GHz and the intensity of mmWave end-users $\lambda_{\text{U}} = 10^4\text{km}^{-2}$ are used in all numerical calculations. 

\section{Performance Analysis in mmWave Network} 
In this section, we derive the analytical expressions for the coverage probability, the throughput corresponding to the fixed-rate coding scheme and the throughput upper bound of mmWave networks.
Note that all the performance metrics are characterized as the polynomials of spatial multiplexing gain $k$. Moreover, the derived polynomial formulas separate the positive and negative effects of the AP densification on the performance of mmWave networks.

\subsection{Coverage Probability Analysis} 
For the mmWave network of relative density $\psi$, the coverage probability at SINR threshold $\tau$ is derived in Theorem \ref{thm: coverage probability for dense mmWave network}, where the expression of $\mathcal{C}_\psi(\tau,k)$ is given in (\ref{eq: simplified coverage probability}). However, it is still obscure how the coverage probability is affected by the relative density $\psi$. 
To provide more insights, we derive the polynomial formula for the coverage probability $\mathcal{C}_{\psi}(\tau,k)$ as follows.
\begin{thm}
	For a mmWave network of relative density $\psi$, the coverage probability $\mathcal{C}_{\psi}(\tau,k)$ at the SINR threshold $\tau$ can be written as a polynomial of the spatial multiplexing gain $k$ as follows:
	\begin{eqnarray}
	\mathcal{C}_{\psi}(\tau,k) 
	&=&
	c_0(\tau,\psi) + 
	 \sum\limits_{l=1}^{\infty} c_l (\tau,\psi) k^l, 
	 \label{eq: SINR interms of polynomial}
	\end{eqnarray}
	where the expression of coefficient $c_0(\tau,\psi)$ is given in (\ref{eq: coefficient c_0}); $c_l(\tau,\psi)$ for $l\in\mathbb{N}^+$ is given in (\ref{eq: coefficient c_l}).
The coverage probability $\mathcal{C}_{\psi}(\tau,k)$ can be further approximated by its $L$ first terms, where
 	\begin{eqnarray}
 \mathcal{C}_{\psi}(\tau,k) \approx  c_0(\tau,\psi) + \sum\limits_{l=1}^{L} c_l (\tau,\psi) k^l  
 \label{eq: coverage probability with finite terms} 
 \end{eqnarray}
 with the approximation error
 \begin{equation}
 \left| \mathcal{C}_{\psi}(\tau, k) - \sum\limits_{l=0}^{L} c_l (\tau,\psi) k^l \right| 
 <\frac{(2^{\mu}-1) e^\psi  \psi^{L+2}}{(L+2)!}.
 \label{eq: approximation error of coverage probability}
 \end{equation}
 \label{thm: taylor expansion of coverage probability}
\end{thm}
\vspace*{-20pt}
\begin{proof}
	Refer to Appendix \ref{proof: taylor expansion of coverage probability}.
\end{proof}
Note that the coefficients $c_0(\tau,\psi), c_l(\tau,\psi)$ are independent of the spatial multiplexing gain $k$. It implies that $k$ is decoupled from $\psi$ and $\tau$ in (\ref{eq: SINR interms of polynomial}).
\textcolor{black}{
For the special case of the spatial multiplexing gain $k=1$, it follows from Corollary \ref{cor: bound of coverage probability} that $\mathcal{C}_{\psi}(\tau, 1)$ provides the upper bound for $\mathcal{C}_{\psi}(\tau, k)$. Thus, we have $\mathcal{C}_{\psi}(\tau,k) \leq \mathcal{C}_{\psi}(\tau, 1) = \sum\limits_{l=0}^{\infty} c_l (\tau,\psi)$.
}

In Theorem \ref{thm: taylor expansion of coverage probability}, the coverage probability of mmWave network is written as a sum of $c_0(\tau,\psi)$ and $\sum\limits_{l=1}^{\infty} c_l (\tau,\psi) k^l$. 
Following the proof of Theorem \ref{thm: taylor expansion of coverage probability}, the term $c_0(\tau,\psi)$ in (\ref{eq: coefficient c_0}) is equivalent to the probability $\mathbb{P}(\text{SNR}>\tau)$, which is monotonically increasing with the relative density $\psi$. That is to say, $c_0(\tau,\psi)$ represents the gain in the coverage probability with the AP density. Moreover, $c_0(\tau,\psi)$ is independent of $k$ since the signal strength at the end-users is not affected by the spatial multiplexing gain of mmWave APs.
The second term $\sum\limits_{l=1}^{\infty} c_l (\tau,\psi) k^l$ in (\ref{eq: SINR interms of polynomial}) captures the loss in the coverage probability due to the aggregate interference of the APs. It follows that the detrimental effect of the AP densification on the coverage probability $\mathcal{C}_{\psi}(\tau,k)$ can be expressed as a polynomial of the spatial multiplexing gain $k$.

\subsection{Throughput Analysis for Fixed-Rate Coding Scheme}
Assume that the fixed-rate coding scheme with the rate threshold $\rho$ is implemented. Then the throughput of mmWave network $\mathcal{T}(\rho; \psi,k)$ is given in (\ref{eq: throughput of fixed-rate coding scheme}).
Note that the coverage probability $\mathcal{C}_\psi(\tau,k)$ is derived in Theorem \ref{thm: taylor expansion of coverage probability}. To obtain the throughput, we then need to derive the bandwidth $W$ for mmWave end-users.

For analytical tractability, a common assumption is that the mmWave AP deploys the round-robin scheduling to allocate the time-frequency resources. 
It follows that the end-user at the origin is assigned with the bandwidth $W = B/\Psi$, where $B$ denotes the total bandwidth available at the mmWave AP at $x_0$ and $\Psi$ represents the total number of end-users that share the radio resources with the end-user at the origin \cite{singh2013offloading,saha2019unified}.

It follows from the minimum path loss association rule that a mmWave AP covers the area within its LOS-ball of radius $R_\text{B}$ and belonging to its Voronoi cell, as shown in Fig.\ref{fig:voronoi_diagram}. 
We remark that the LOS-ball model limits the coverage of a single mmWave AP, which differentiates the AP coverage in mmWave bands from the typical Voronoi cell in \cite{singh2013offloading}.
Given the spatial multiplexing gain $k$, the coverage area of the mmWave AP is then divided into $k$ sectors, as illustrated in Fig.\ref{fig:spatial multiplexing gain}. 
It follows that $\Psi$ refers to the number of end-users which are covered by the mmWave AP at $x_0$ and within the same sector of $\frac{2\pi}{k}$ angle width as the origin.
Next, we present the main result on the bandwidth distribution $W = B/\Psi$ for the end-user at the origin.
\begin{thm}
	Consider a mmWave network of relative density $\psi$, where the spatial multiplexing gain of APs is $k$. Then the bandwidth assigned to the end-user at the origin equals to $W =  B/\Psi$, where $\Psi$ has the probability mass function (PMF) 
	\begin{equation}
	\mathbb{P}\left(\Psi = n\right) =\mathcal{K}_\text{T}(n,k; \psi,\lambda_{\text{U}}), \; n \in\mathbb{N}^+,
	\label{eq: PMF tagged AP sector}
	\end{equation}
	and the expression of $\mathcal{K}_\text{T}(\cdot)$ is shown in (\ref{eq: tagged end-user load}).
	Since $\psi \geq 1$, the PMF of $\Psi$ can be further simplified as the expression in (\ref{eq: simplied tagged load}).
	\label{thm: PMF for tagged AP sector}
\end{thm}

\begin{proof}
	Refer to Appendix \ref{proof: bandwidth distribution}.
\end{proof}
 Following Theorem \ref{thm: PMF for tagged AP sector}, the average bandwidth $\overline{W}$ for the mmWave end-user at the origin is derived as follows.
\begin{cor}
 In the mmWave network of relative density $\psi$ and spatial multiplexing gain $k$, the average bandwidth assigned to the end-user at the origin is given as
 \begin{equation}
 \overline{W} = \frac{B}{\mathbb{E}[\Psi]} = \frac{B}{1 + \xi\frac{\pi R_\text{B}^2 \lambda_{\text{U}}}{k \psi}},
 \label{eq: average bandwidth}
 \end{equation}
 where $B$ is the total bandwidth available at the mmWave AP; $\xi = 1.28$ is the bias factor for the end-user at the origin. 
 \label{cor: average bandwidth}
\end{cor}
\begin{proof}
Following the similar steps in \cite{singh2013offloading}, we have ${\mathbb{E}[\Psi]} = 1 + \xi\frac{\pi R_\text{B}^2 \lambda_{\text{U}}}{k \psi}$ with the bias factor $\xi$, where $\xi = 1.28$ is the bias factor for the Voronoi cell containing the origin. For the Voronoi cell that does not contain the origin, we have the bias factor $\xi = 1$.
\end{proof}
Since we have $\xi\frac{\pi R_\text{B}^2 \lambda_{\text{U}}}{k \psi} \gg 1$, Corollary \ref{cor: average bandwidth} then implies that the average bandwidth $\overline{W}$ of the mmWave end-user at the origin scales linearly with the relative density $\psi$.
It implies that as the density of APs increases, the bandwidth gain of end-users is the key enabler for increasing the throughput of mmWave network.

Based on the coverage probability derived in Theorem \ref{thm: taylor expansion of coverage probability} and the bandwidth distribution derived in Theorem \ref{thm: PMF for tagged AP sector}, we then present the main result on the throughput of mmWave networks.
\begin{thm}
	Consider the mmWave network of relative density $\psi$, where the spatial multiplexing gain of the mmWave AP equals to $k$. Then the throughput achieved by the fixed-rate coding scheme with the rate threshold $\rho$ can be given as
	\begin{equation}
	\mathcal{T}(\rho; \psi,k) 
	= \lambda_{\text{U}} \rho \mathbb{E}_{\Psi}\left[ c_0 \left(2^{\rho \Psi/ B} - 1,\psi \right) \right] + \lambda_{\text{U}} \rho \sum_{l=1}^{\infty} k^l\mathbb{E}_{\Psi}\left[ c_l \left(2^{\rho \Psi/ B} - 1,\psi \right) \right],
	\label{eq: rate coverage}
	\end{equation}
	where $B$ is the total bandwidth available at the AP; 
	the PMF of $\Psi$ is derived in (\ref{eq: PMF tagged AP sector});
	the expression of $c_l(\tau,\psi)$ for $l\in\mathbb{N}$ is given in Theorem \ref{thm: taylor expansion of coverage probability}.
	\label{thm: rate coverage}	
\end{thm}
\begin{proof}
	Given the bandwidth distribution $W = B/\Psi$ in Theorem \ref{thm: PMF for tagged AP sector}, the throughput under the fixed-rate coding scheme in (\ref{eq: throughput of fixed-rate coding scheme}) can then be written as
	\begin{eqnarray}
	 \mathcal{T}(\rho; \psi,k)  &=& \lambda_{\text{U}} \rho \mathbb{P}\left(\frac{B}{\Psi}\log_2(1+\text{SINR})>\rho\right) \nonumber\\
	&=& \lambda_{\text{U}} \rho \mathbb{E}_{\Psi}\left[ \mathcal{C}_{\psi}(2^{\rho \Psi/B} - 1,k) \right] 
	\label{eq: rate coverage in coverage probability},
	\end{eqnarray}
	where $\mathcal{C}_{\psi}(\tau, k)$ is the coverage probability of the mmWave network at the SINR threshold $\tau = 2^{\rho n/B } - 1$ for $\Psi = n$.
	The throughput in (\ref{eq: rate coverage}) can be immediately obtained from the coverage probability in (\ref{eq: SINR interms of polynomial}) and the PMF of $\Psi$ in (\ref{eq: PMF tagged AP sector}).
\end{proof}
It follows from (\ref{eq: rate coverage}) that the throughput of mmWave network can be written as the summation of two terms.
As proved in Theorem \ref{thm: taylor expansion of coverage probability}, $c_0(\tau,\psi)$ in the first term denotes the probability $\mathbb{P}(\text{SNR} > \tau)$ and is a monotonically increasing function of $\psi$, since the desired signal strength increases with the AP density. 
The second term of (\ref{eq: rate coverage}) corresponds to the contribution of the aggregate interference received by the end-user, which captures the detrimental effect of the AP densification on the throughput of mmWave network. 

It follows from Theorem \ref{thm: PMF for tagged AP sector} that the bandwidth for the end-user, $B/\Psi$, tends to increase with $\psi$. Thus the densification of APs can lower the SINR threshold $\tau = 2^{\rho \Psi/ B} - 1$ in (\ref{eq: rate coverage}). In other words, the variable $\tau = 2^{\rho \Psi/ B} - 1$ in (\ref{eq: rate coverage}) characterizes the throughput gain provided by the increase in the end-user's bandwidth as the AP density increases.

\subsection{Throughput Upper Bound Characterization}
While the throughput achieved by the fixed-rate coding scheme will vary with different choices of rate threshold $\rho$, the throughput upper bound is irrelevant to $\rho$. Given the instantaneous SINR, the maximum achievable rate for the end-user is set by the Shannon bound, which equals to $W\log_2(1 +\text{SINR})$. It follows that the upper bound of the network throughput is achieved if each end-user can reach its instantaneously maximum achievable rate \cite{alammouri2018sinr}. In the following Theorem, we derive the upper bound of the throughput for the mmWave network.
\begin{thm}
	Consider the mmWave network of relative density $\psi$, where the spatial multiplexing gain of APs equals to $k$. By assuming that the interference is treated as the noise, the upper bound of the throughput is given as
	\begin{equation}
	\overline{\mathcal{T}}\left(\psi,k \right)
	=
	\lambda_{\text{U}}\mathbb{E}_{\Psi}\left[\frac{B}{\Psi}\right]
	\log_2(e)
	\left(
	\int_{0}^{\infty} \frac{c_0 (\tau,\psi)}{1+\tau} \text{d}\,\tau
	+
	\sum\limits_{l=1}^{\infty}  k^{l} 
	\int_{0}^{\infty} \frac{c_l (\tau,\psi)}{1+\tau} \text{d}\,\tau
	\right) .
	\label{eq: bounds on throughput}
	\end{equation}
	Here,
	the PMF of $\Psi$ is derived in Theorem \ref{thm: PMF for tagged AP sector}; $c_l(\tau,\psi)$ is given in Theorem \ref{thm: taylor expansion of coverage probability}.
	\label{thm: network throughput}	
\end{thm}
\begin{proof}
	See Appendix \ref{proof: throughput upper bound}.
\end{proof}

We remark that the coverage probability in (\ref{eq: SINR interms of polynomial}), the throughput in (\ref{eq: rate coverage}), and the throughput upper bound in (\ref{eq: bounds on throughput}) decouple the improvement and degradation of the network performance as the mmWave AP density increases. Before delving into the asymptotic trends in the dense AP deployments, we want to emphasize that the full buffer is assumed for all the performance analysis in this section. It implicitly means that to apply the results in (\ref{eq: SINR interms of polynomial}), (\ref{eq: rate coverage}) and (\ref{eq: bounds on throughput}), the density of mmWave end-users $\lambda_{\text{U}}$ needs to be sufficiently large, which can guarantee that no AP sectors are idle.

\section{Effect of AP Densification on mmWave Network} 
We now investigate how the throughput of a mmWave network scales with the AP density, where the implementation of the fixed-rate coding scheme is assumed throughout this section.
Note that the scaling law of throughput is different in the power-limited and interference-limited scenarios.
Therefore, we first derive the range of relative density threshold $\psi^*$, below which the mmWave network is power-limited and beyond which is interference-limited.
Then, we introduce the concept of densification gain, which captures the throughput gain as the density of APs increases.
Our goal here is to illustrate the presence of the densification plateau for the mmWave network without spatial multiplexing. Moreover, we demonstrate that such a densification plateau can be overcome by employing the spatial multiplexing at the mmWave APs.

\subsection{Relative Density Threshold for mmWave Network} 
For the mmWave network, let $\psi^*(\tau,k)$ denote the relative density that maximizes the coverage probability i.e. $\psi^*(\tau,k) = \text{arg}\max\limits_{\psi\in\mathbb{R}^+} \mathcal{C}_\psi(\tau, k)$, where $\psi^*(\tau,k)$ is defined as the relative density threshold.
Intuitively, the mmWave network with the relative density $\psi < \psi^*(\tau,k) $ is in the power-limited regime, where the coverage probability is limited by the signal strength due to the lack of AP coverage. It follows that the coverage probability will benefit from increasing the AP density, or equivalently, increasing $\psi$.
However, as $\psi$ increases, the mmWave network will transit from the power-limited regime into the interference-limited regime, where the coverage probability starts to decay with $\psi$.
In \cite{bai2015coverage}, $\psi^* = 4$ is shown to be the empirically relative density threshold. 
Next, we analytically derive the range of the relative density threshold.

\begin{lem}
	For a mmWave network, the range of relative density threshold $\psi^*(\tau,k)$ satisfies $2 \leq \psi^*(\tau,k) \leq 4$ for all the SINR threshold $\tau \leq 15$dB and spatial multiplexing gain $k\leq \frac{2\pi}{\theta_{\text{A}}}$.
	\label{lem: optimal relative density}
\end{lem}

\begin{proof}
	To prove the range of $\psi^*(\tau,k)$, we need to show that $\mathcal{C}_{\psi}(\tau,k)$ increases with the relative density $\psi$ for $\psi<2$, while decreases with $\psi$ for $\psi>4$.
	For the mmWave network with the relative density $\psi < 2$, the fact $|\sum\limits_{l=1}^{\infty} c_l (\tau,\psi) k^l |<0.1, \forall \tau\leq15\text{dB, }\, \forall k \leq \frac{2\pi}{\theta_{\text{A}}},$ demonstrates that the interference has a negligible impact on the coverage probability $\mathcal{C}_{\psi}(\tau,k)$ in (\ref{eq: SINR interms of polynomial}) within that range of relative density. It also indicates SINR $\approx$ SNR when the mmWave network has a relative density $\psi < 2$. As the desired signal strength increases with the AP density, $\mathcal{C}_{\psi}(\tau,k)$ increases with $\psi$ if $\psi<2$. 
	However, if $\psi > 4$, we have $c_0(\tau,\psi) > 0.95,  \forall \tau\leq15\text{dB}, \forall  k \leq \frac{2\pi}{\theta_{\text{A}}}$, which implies that only marginal gain in SNR can be obtained beyond the relative density $\psi = 4$ and thus the interference determines the coverage probability. It follows that the mmWave network is in the interference-limited region and $\mathcal{C}_{\psi}(\tau,k)$ decreases with $\psi$ if $\psi > 4$.
\end{proof}
\begin{figure}[!t]
	\vspace{-0.1cm}
	\centering
	\includegraphics[width=0.5\textwidth]{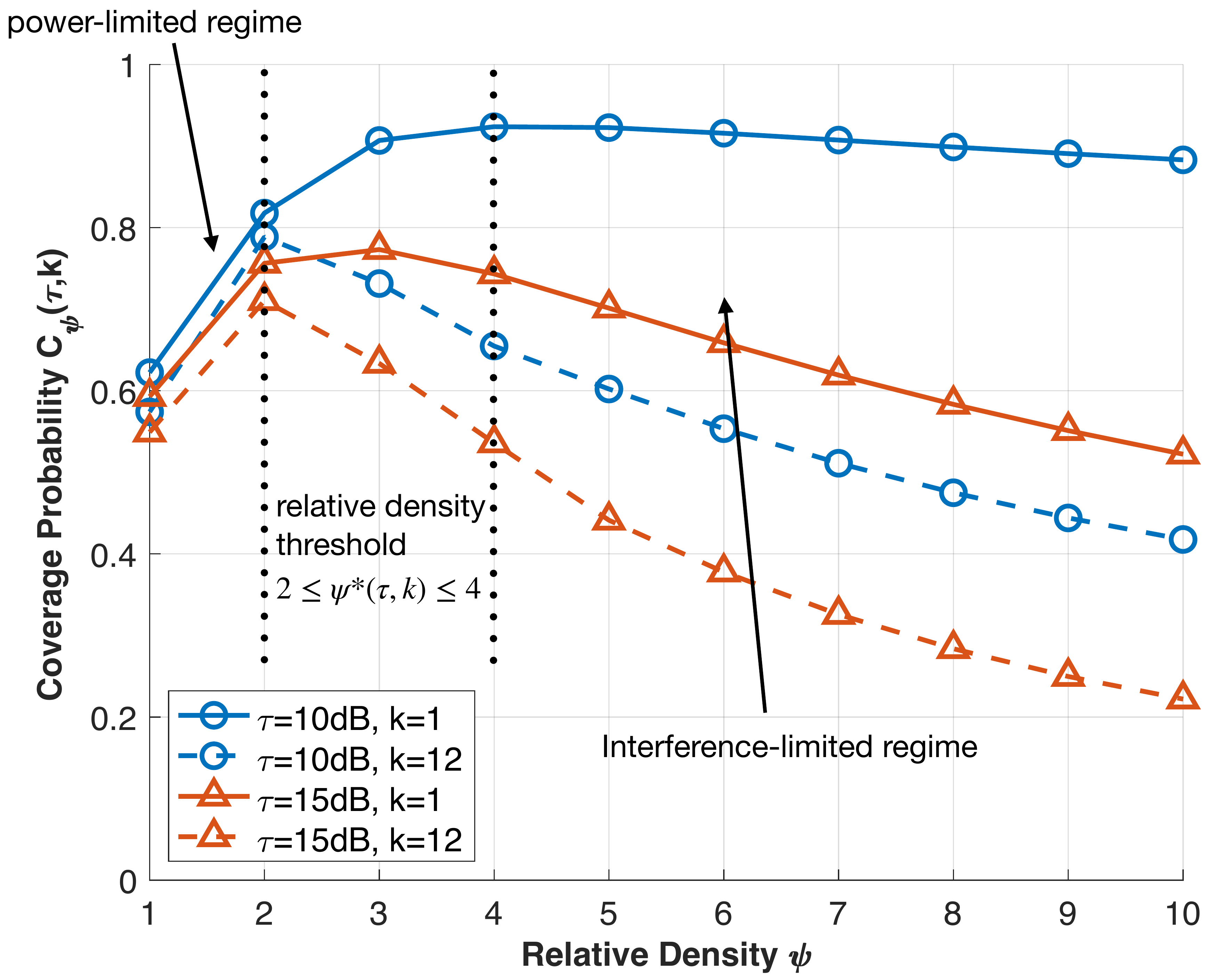}
	\caption{The variation of coverage probability with relative density $\psi$ for mmWave network. For illustration purpose, the SINR threshold is selected as $\tau=10$dB and $\tau = 15$dB, which satisfies the SINR threshold condition i.e. $\tau \leq 15$dB in Lemma \ref{lem: optimal relative density}. It can be observed that $\mathcal{C}_\psi(\tau,k)$ increases with $\psi$ in the power-limited regime i.e. $\psi < 2$ and decay with $\psi$ in the interference-limited regime i.e. $\psi > 4$.}
	\label{fig:optimal_density}
	\vspace{-20pt}
\end{figure}
For a mmWave network, the relative density threshold $\psi^*(\tau,k)$ corresponds to the AP density that is sufficiently large in terms of the desired signal strength, yet not over-densifying in terms of the aggregate interference.
In Fig.\ref{fig:optimal_density}, we numerically demonstrate the range of $\psi^*(\tau,k)$ by plotting the variation of coverage probability $\mathcal{C}_{\psi}(\tau,k)$ with the relative density $\psi$. Here, the SINR thresholds are chosen as $\tau = 10$dB and $\tau = 15$dB for illustration purposes.
It is observed that for the relative density $\psi < 2$, the coverage probability will always benefit from increasing $\psi$. This is because the coverage probability $\mathcal{C}_\psi(\tau,k)$ with the relative density $\psi < 2$ is limited by the desired signal strength for all the SINR threshold $\tau < 15$dB. It follows that the desired signal strength increases with the AP density, which results in a significant improvement in the coverage probability.
However, for the mmWave network of relative density $\psi\geq4$, the interference dominates the coverage probability. Therefore, increasing $\psi$ will harm the coverage probability due to the increase in the interfering APs.
Fig.\ref{fig:optimal_density} also numerically demonstrates that with the SINR threshold $\tau \leq 15\text{dB}$, the relative density threshold $\psi^*(\tau,k)$ ranges from 2 to 4, which agrees with the statement in Lemma \ref{lem: optimal relative density}.

\subsection{Densification Gain}
Next, we introduce the densification gain to study the effect of AP densification on the throughput of mmWave networks. 
For a mmWave network of relative density $\psi$, the densification gain is defined as the following ratio,
\begin{eqnarray}
	\gamma(\psi, k) \triangleq  \frac{\mathcal{T}(\rho(\psi); \psi, k)}{\mathcal{T}(\rho_0; 1, k)}
	= \frac{\mathcal{T}(\psi\rho_0; \psi, k)}{\mathcal{T}(\rho_0; 1, k)},
	\label{def: densification gain}
\end{eqnarray}
where $\rho_0$ denotes the rate threshold when the mmWave network has the relative density $\psi = 1$; the throughput $\mathcal{T}(\rho(\psi); \psi, k)$ with the rate threshold $\rho(\psi)$ is given in (\ref{eq: rate coverage}). Note that we choose the rate threshold $\rho(\psi) = \psi\rho_0$. Correspondingly, by assuming that the interference can be completely eliminated, the throughput of mmWave network $\lambda_{\text{U}}\rho(\psi)$ can scale linearly with $\psi$. 

Note that the value of densification gain $\gamma(\psi, k)$ characterizes the relationship between the AP density and the throughput of mmWave network as follows. If the densification gain $\gamma(\psi,k)$ equals to $\psi$, then it follows from (\ref{def: densification gain}) that $\mathcal{T}(\psi\rho_0; \psi, k) = \psi \mathcal{T}(\rho_0; 1, k)$, which means that the throughput scales linearly with the relative density $\psi$. For the case that $\gamma(\psi,k)$ is upped bounded by a finite constant denoted by $a$, the throughput $\mathcal{T}(\psi\rho_0; \psi, k)$ converges to a finite number $a \mathcal{T}(\rho_0; 1, k)$, which illustrates the existence of densification plateau for mmWave networks.

Next, we derive the densification gain $\gamma(\psi, k)$ for different relative densities $\psi$. Following Lemma \ref{lem: optimal relative density}, the relative density $\psi < \psi^*(\tau,k)$ and $\psi > \psi^*(\tau,k)$ correspond to the power-limited regime and the interference-limited regime, respectively. We start to derive the densification gain $\gamma(\psi,k)$ for the power-limited mmWave network i.e. $\psi < 2$.
\begin{lem}
	For the mmWave network of the relative density $\psi < 2$, the densification gain $\gamma(\psi,k)$ satisfies the inequality $\gamma(\psi,k) \geq \psi$.
	\label{lem: densification gain - power-limited regime}
\end{lem}
\begin{proof}
	It is equivalently to prove $\mathcal{T}(\psi \rho_0; \psi, k) \geq \psi\mathcal{T}(\rho_0; 1, k)$. By approximating the bandwidth $W$ by the average bandwidth in (\ref{eq: average bandwidth}), it follows from Lemma \ref{lem: optimal relative density} that the throughput $\mathcal{T}(\psi \rho_0; \psi, k) = \lambda_\text{U}\psi\rho_0\mathcal{C}_\psi(2^{\rho_0/kW_0} - 1, k) \geq \psi \lambda_{\text{U}}\rho_0\mathcal{C}_1(2^{\rho_0/kW_0} - 1, k) = \psi\mathcal{T}(\rho_0; 1,k)$ with $W_0 = \frac{\xi\pi R_\text{B}^2\lambda_{\text{U}}}{B}$ for all $\psi < 2$ if the SINR threshold $\tau = 2^{\rho_0/kW_0} - 1 < 2^{\rho_0/W_0} - 1 < 15$dB. Note that for the mmWave network with the large bandwidth available, $\tau < 15$dB is considered to be a valid assumption \cite{bai2015coverage,rappaport2013broadband,rappaport2015wideband}. The Lemma then follows.
\end{proof}

Lemma \ref{lem: densification gain - power-limited regime} implies that for the power-limited mmWave network, the throughput increases at least linearly with the AP density.
Next, we move to the interference-limited mmWave network when the relative density $\psi$ becomes large. We first obtain the asymptotic result on the densification gain for the mmWave network without spatial multiplexing i.e. $k=1$.
\begin{lem}
	For the mmWave network with the spatial multiplexing gain $k=1$, the densification gain $\gamma(\psi,1)$ is upper bounded by a finite constant as the relative density $\psi \rightarrow +\infty$.
	\label{lem: plateau for k=1}
\end{lem}
\begin{proof}
See Appendix \ref{proof: lem - densification plateau}.
\end{proof}
Since the densification gain $\gamma(\psi,1)$ is saturated to a finite constant as $\psi$ increases, the existence of densification plateau for $k=1$ is then proved in Lemma \ref{lem: plateau for k=1}.
 In the following, we continue to show that the densification gain $\gamma(\psi,k)$ is proportional to $\psi$ if the spatial multiplexing gain $k$ increases linearly with $\psi$.
\begin{lem}
	For the interference-limited mmWave network, the densification gain $\gamma(\psi,k)$ is proportional to the relative density $\psi$ if the spatial multiplexing gain $k$ increases linearly with $\psi$ as $\psi \rightarrow +\infty$.
	\label{lem: densification gain with increased k}
\end{lem}
\begin{proof}
Following the similar steps in (\ref{proofeq: limitation of SIR - upper}) of Appendix \ref{proof: lem - densification plateau}, we then have
\begin{eqnarray}
\lim\limits_{\lambda_{\text{A}} \rightarrow +\infty} \lambda_{\text{A}} \text{SINR} 
&\geq& 
\frac{ 1 }{2 \pi  \sigma R_\text{B}^{\alpha_{\text{L}}}},
\label{proofeq: limitation of SIR - lower}
\end{eqnarray}
where $\sigma = \int_{0}^{R_\text{B}} r^{-\alpha_\text{L}+ 1}  \text{d}r$ is a non-trivial finite constant. 

It follows from Lemma \ref{lem: densification gain - power-limited regime} that $\mathcal{T}(\psi\rho_0; \psi, k) \approx \lambda_{\text{U}}\psi\rho_0\mathcal{C}_\psi(\tau,k)$ at $\tau = 2^{\frac{1}{k}\cdot \frac{\rho_0\xi\pi R_\text{B}^2\lambda_{\text{U}}}{B}}-1$. By letting $\tau_0 = \frac{1}{2\lambda_{\text{A}}}   \cdot \frac{1}{2 \pi  \sigma R_\text{B}^{\alpha_{\text{L}}}}$, we can conclude from (\ref{proofeq: limitation of SIR - lower}) that $\lim\limits_{\lambda_{\text{A}} \rightarrow +\infty} \text{SINR} \geq  \lim\limits_{\lambda_{\text{A}} \rightarrow +\infty} \frac{1}{\lambda_{\text{A}}} \cdot \frac{1}{2 \pi  \sigma R_\text{B}^{\alpha_{\text{L}}}} > \tau_0$. Note that $ \psi = \pi R_\text{B}^2 \lambda_\text{A}$ and $R_\text{B}$ is unchanged, thus $\lim\limits_{\psi \rightarrow +\infty}\text{SINR} = \lim\limits_{\lambda_{\text{A}} \rightarrow +\infty}\text{SINR}> \tau_0$ a.s.. It means that $\mathcal{C}_\psi(\tau_0,k) = \mathbb{P}(\text{SINR} > \tau_0)$ is a finite non-trial number if $\psi$ is sufficiently large. It follows that if $\lim\limits_{\psi \rightarrow +\infty}\tau= \tau_0$, then we have $\lim\limits_{\psi \rightarrow +\infty} \mathcal{C}_\psi(\tau,k) = \mathcal{C}_\psi(\tau_0,k)$. Note that $\lim\limits_{\psi \rightarrow +\infty}\tau= \lim\limits_{\psi \rightarrow +\infty} \; 2^{\frac{1}{k}\cdot \frac{\rho_0\xi\pi R_\text{B}^2\lambda_{\text{U}}}{B}}-1$ and $\tau_0 = \frac{1}{\psi} \cdot \frac{R_\text{B}^{-\alpha_{\text{L}} + 2}}{4\sigma}$. By using the fact that $\log_2(1 + 1/\psi) \approx 1/\psi$ as $\psi \rightarrow +\infty$, the proof then completes.
\end{proof}
Lemma \ref{lem: densification gain with increased k} implies that given the rate threshold $\rho$, the increased spatial multiplexing gain $k$ improves the coverage probability $\mathcal{C}_\psi(\tau,k)$ by lowering the SINR threshold $\tau$, which consequently drives the gain in the throughput $\mathcal{T}(\rho; \psi, k)$.
Note that an implicit assumption on the mmWave AP is that the spatial multiplexing gain satisfies $k\leq \frac{2\pi}{\theta_{\text{A}}}$, where $\theta_{\text{A}}$ is the main lobe width of the AP. Therefore, as $k$ increases with $\psi$, $\theta_{\text{A}}$ also needs to decrease with $\psi$.
\begin{figure}[!t]
	\vspace{-0.1cm}
	\begin{minipage}[t]{0.47\textwidth}
		\centering
		\includegraphics[width=3in]{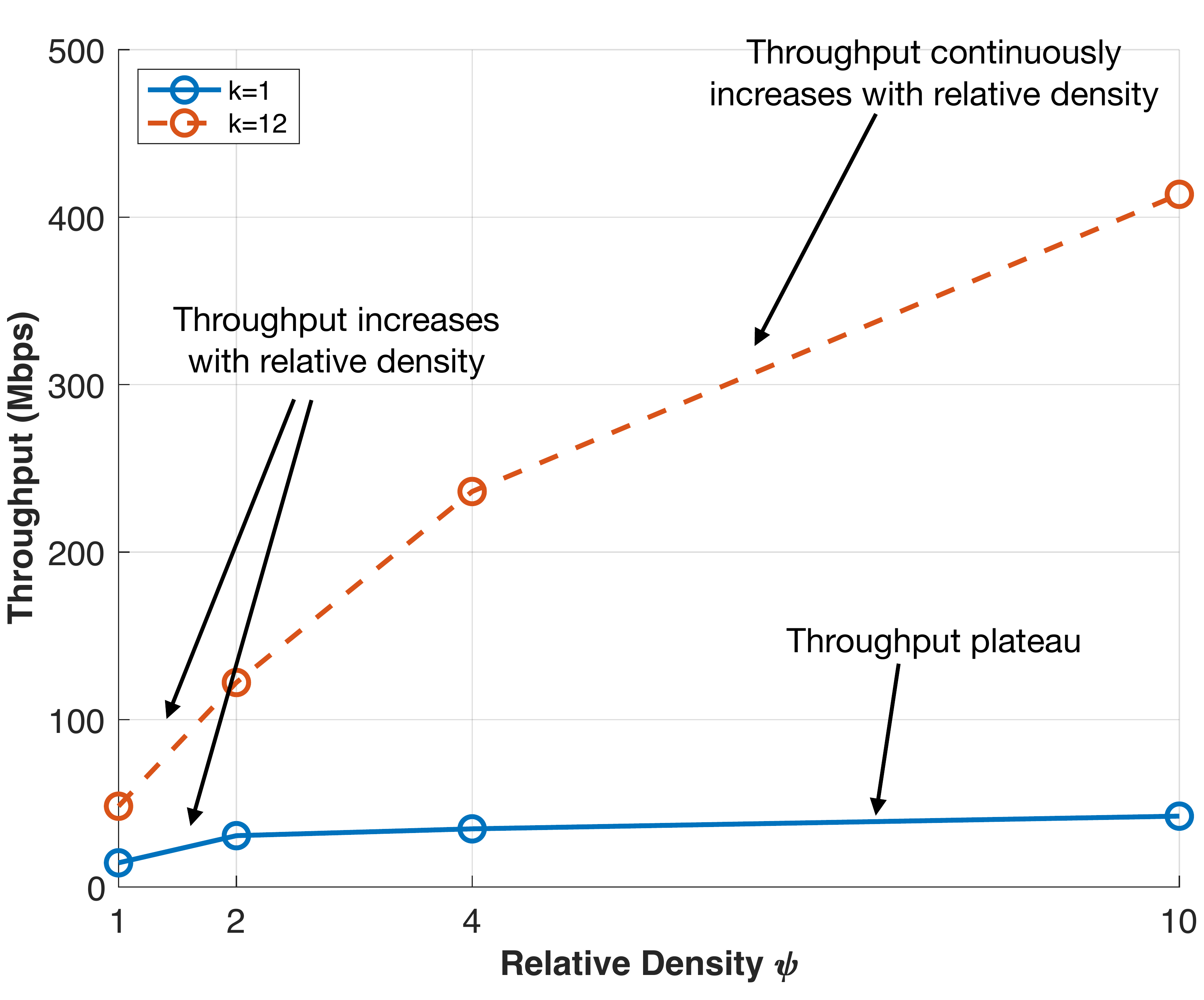}
		\caption{The throughput of mmWave network (normalized by $\lambda_{\text{U}}$), i.e., $\mathcal{T}(\rho(\psi); \psi,k)/\lambda_{\text{U}}$, with spatial multiplexing gain $k=1$ and $k=12$. The rate threshold of fixed-rate coding scheme is chosen as $\rho(\psi) = \psi\rho_0$ with $\rho_0 = 80$Mbps.}
		\label{fig:achievable_throughput}
	\end{minipage} %
	\hfill
	\begin{minipage}[t]{0.47\textwidth}
		\centering
		\includegraphics[width=3in]{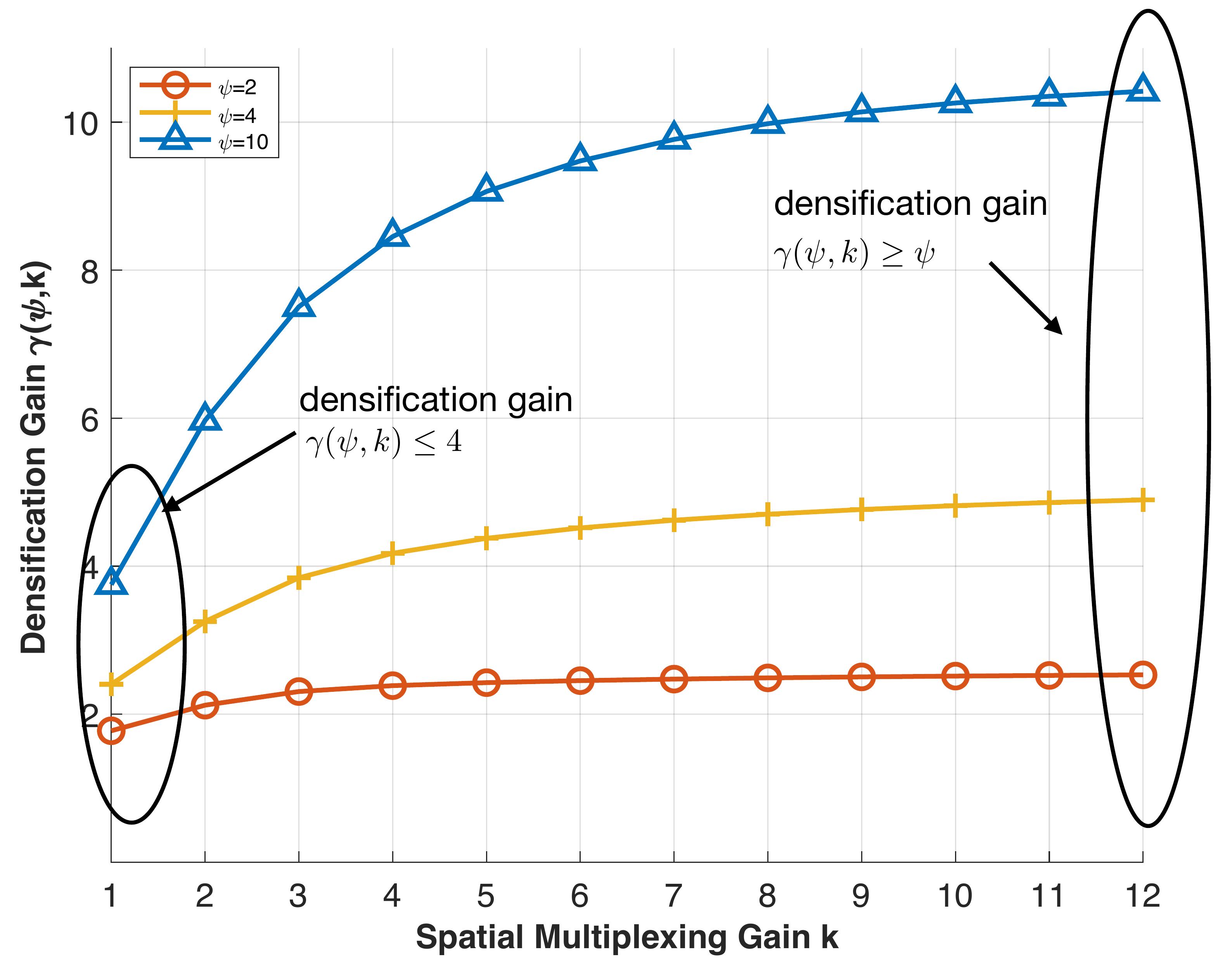}
		\caption{Densification gain $\gamma(\psi,k)$ vs. spatial multiplexing gains $k$ ($\rho_0 = $80Mz). Without spatial multiplexing i.e. $k=1$, the mmWave network reaches the densification plateau at $\psi = 4$ since $\gamma(\psi,1) \leq 4$. With $k\geq 8$, $\gamma(\psi,k) \geq \psi$ means the throughput scales at least linearly with $\psi$.}
		\label{fig:densification_gain}
	\end{minipage}
	\vspace{-20pt}
\end{figure}

In Fig.\ref{fig:achievable_throughput}, we compare the throughput normalized by $\lambda_{\text{U}}$, i.e., $\mathcal{T}(\rho(\psi); \psi, k)/\lambda_{\text{U}}$ for the mmWave network with different relative densities $\psi$. 
It can be seen from Fig.\ref{fig:achievable_throughput} that for the spatial multiplexing gain $k=1$, the AP densification beyond $\psi = 4$ does not provide any further improvement for the throughput of mmWave network, which illustrates that the mmWave network with $k=1$ reaches the densification plateau at $\psi = 4$. 
In sharp contrast to $k=1$, the throughput of the mmWave network with $k=12$ can be improved by increasing the AP density for $\psi \leq 10$, which implies that the deployment of spatial multiplexing at mmWave APs overcomes the densification plateau for the mmWave network.

In Fig.\ref{fig:densification_gain} we show the variation of densification gain $\gamma(\psi,k)$ with $k$ for different relative densities $\psi$.
From Fig.\ref{fig:densification_gain}, it is clear that for the relative density $\psi = 2$, the densification gain $\gamma(\psi,k) \geq \psi$ holds for all $k$, as stated in Lemma \ref{lem: densification gain - power-limited regime}.
As the relative density $\psi$ increases from $2$ to $4$, the mmWave network transits from the power-limited region to the interference region.
Then the densification gain $\gamma(\psi,k)$ for $k>1$ diverges from $\gamma(\psi,1)$.

Fig. \ref{fig:densification_gain} numerically demonstrates Lemma \ref{lem: plateau for k=1} and \ref{lem: densification gain with increased k}.
It can be observed that with the spatial multiplexing gain $k=1$, the densification gain $\gamma(\psi,1)$ for different relative densities $\psi$ are clustered to a constant $a\leq 4$. It indicates that with $k=1$,  only marginal throughput gain can be observed beyond $\psi = 4$. In other words, the mmWave network with $k=1$ reaches the densification plateau at the relative density $\psi = 4$.
In contrast to $k=1$, for $k>1$, the densification gain $\gamma(\psi,k)$ scales with $\psi$, which illustrates that the throughput can consistently increase with the AP density. It is equivalent to say that the densification plateau of the mmWave network can be avoided by deploying the spatial multiplexing gain $k\propto \psi$ at the APs. 

\section{Numerical Results on mmWave Network Throughput}
This section shows the numerical results on throughput gains as the spatial multiplexing gain at mmWave APs increases.
Note that by choosing the optimal rate threshold $\rho^*$, the fixed-rate coding scheme provides the tightest lower bound for the throughput of mmWave networks, while the throughput upper bound is derived in (\ref{eq: bounds on throughput}). 
Moreover, we compare the multi-rate coding scheme with the throughput upper and lower bounds, which motivates a performance-complexity trade-off for coding schemes in the mmWave network.

\subsection{Throughput Gain via Spatial Multiplexing}
We start to provide an intuitive explanation on how the spatial multiplexing gain affects the throughput of mmWave network. 
Corollary \ref{cor: average bandwidth} shows that the average bandwidth of end-user scales linearly with $k$.
Consequently, increasing $k$ equivalently lowers the SINR threshold $\tau = 2^{\rho \Psi/B} - 1$ for the coverage probability $\mathcal{C}_\psi(\tau,k)$ in (\ref{eq: rate coverage in coverage probability}) and thus improves the network throughput $\mathcal{T}(\rho;\psi,k)$.
For completeness, the minor drawback of increasing $k$ is shown in Corollary \ref{cor: bound of coverage probability} as the coverage probability $\mathcal{C}_\psi(\tau,k)$ decreasing with $k$. 

From the analytical viewpoint, by fixing the relative density $\psi$, the throughput in (\ref{eq: rate coverage}) is a polynomial of the spatial multiplexing gain $k$, while the polynomial's coefficient $\mathbb{E}_{\Psi}\left[ c_l \left(2^{\rho \Psi/ B} - 1,\psi \right)\right]$ and $k$ are also correlated due to the dependency of $\Psi$ on $k$, as shown in Theorem \ref{thm: PMF for tagged AP sector}.
To numerically demonstrate the effect of spatial multiplexing on the throughput of mmWave networks, $\mathcal{T}(\rho; \psi,k)/\lambda_{\text{U}}$ in (\ref{eq: rate coverage}) is plotted against $\rho$ for different spatial multiplexing gains $k$ in Fig.\ref{fig: throughput vs. rate -fixed rate coding scheme}. Here, the relative density of mmWave network is chosen as $\psi = 4$ and $\psi = 10$.
From Fig.\ref{fig: throughput vs. rate -fixed rate coding scheme}, it can be clearly seen that the throughput $\mathcal{T}(\rho; \psi, k)$ always increases with the spatial multiplexing gain $k$ regardless of the chosen $\psi$ and $\rho$. Hence, we can conclude that the deployment of spatial multiplexing at mmWave APs is always beneficial for the throughput of mmWave networks.
\begin{figure}[!t]
	\vspace{-0.1cm}
	\centering
	\includegraphics[width=0.5\textwidth]{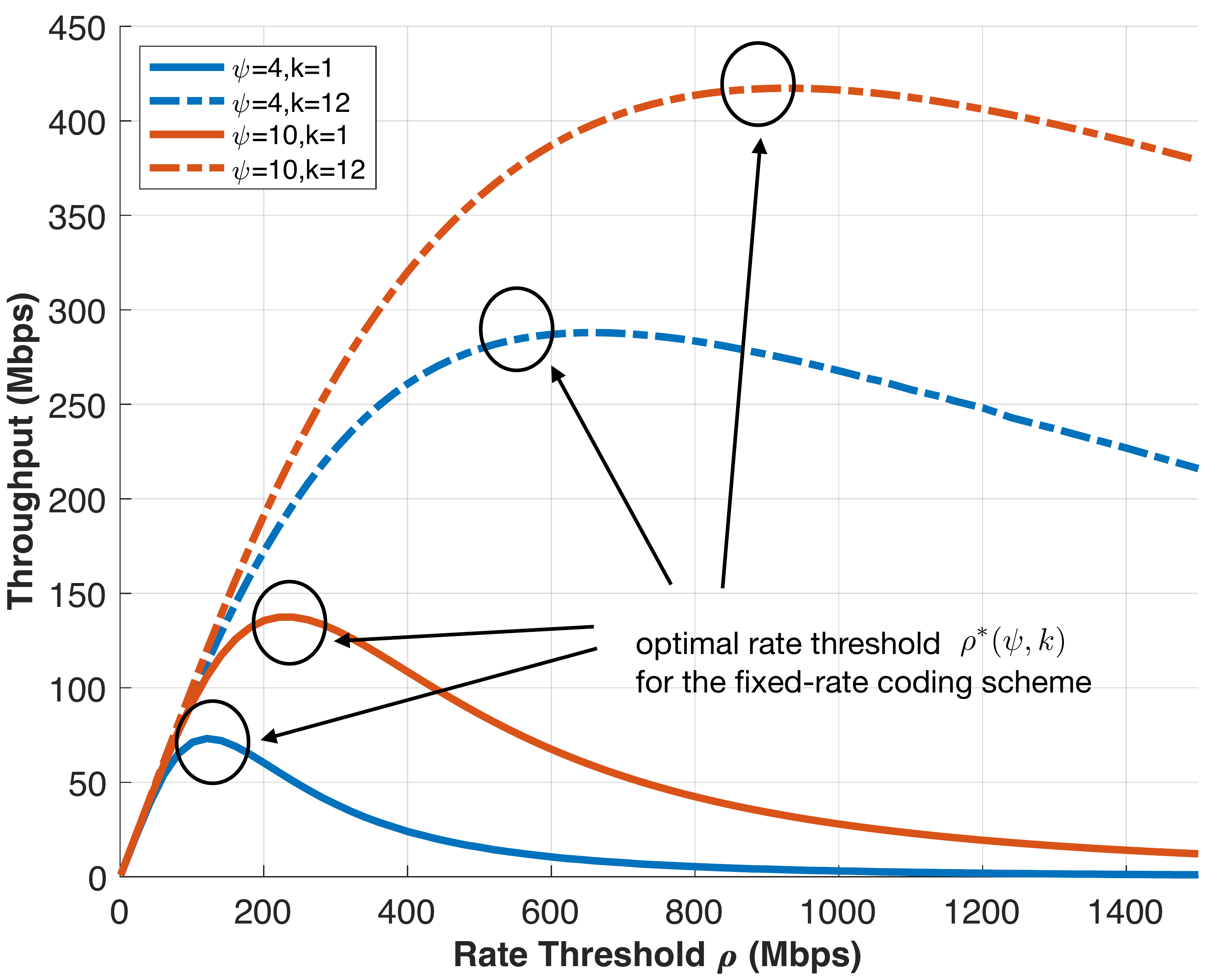}
	\caption{Comparison of the throughput $\mathcal{T}(\rho; \psi,k)/\lambda_{\text{U}}$ for the spatial multiplexing gain $k=1$ and $k=12$. Here, relative density $\psi = 4$ (in red) and $\psi = 10$ (in blue) represent the mmWave network with dense APs and ultra-dense APs, respectively. By deploying the spatial multiplexing gain $k=12$ at mmWave APs, the improvement in the throughput is clearly seen. Note that the fixed-rate coding scheme with the optimal $\rho^*(\psi,k)$ is used to provide the lower bound for mmWave network throughput.}
	\label{fig: throughput vs. rate -fixed rate coding scheme}
	\vspace{-20pt}
\end{figure}

\subsection{Throughput Upper and Lower Bounds}
By deploying the fixed-rate coding scheme, the throughput $\mathcal{T}(\rho; \psi, k)$ of the mmWave network achieves its maximum value when the rate threshold $\rho$ is chosen such that $\rho^*(\psi,k) = \text{arg}\max\limits_{\rho\in\mathbb{R}^+}\,\mathcal{T}(\rho; \psi, k)$. 
In Fig.\ref{fig: throughput vs. rate -fixed rate coding scheme}, we show the variation of throughput normalized by $\lambda_{\text{U}}$, i.e., $\mathcal{T}(\rho; \psi, k)/\lambda_{\text{U}}$ with the rate threshold $\rho$. It can be seen that with the same relative density $\psi$, the optimal rate threshold $\rho^*(\psi,k)$ increases as the spatial multiplexing gain $k$ increases. If $k$ is fixed, then the mmWave network with higher relative density $\psi$ requires a larger rate threshold $\rho^*(\psi,k)$ to maximize the throughput $\mathcal{T}(\rho; \psi, k)$.
By choosing the optimal rate threshold $\rho^*(\psi,k)$, the fixed-rate coding scheme then provides the tightest lower bound for the throughput of mmWave network.

 In Fig.\ref{fig:NT_vs_k}, we plot the throughput upper bound $\overline{\mathcal{T}}(\psi, k)/\lambda_\text{U}$ in (\ref{eq: bounds on throughput}) and the tightest lower bound $\mathcal{T}(\rho^*; \psi, k)/\lambda_{\text{U}}$ for the mmWave network with different spatial multiplexing gains. Here, we choose the relative densities as $\psi = 4$ and $\psi = 10$, which correspond to the dense and ultra-dense deployments of mmWave APs, respectively.
It can be seen in Fig.\ref{fig:NT_vs_k} that both the upper bound and the tightest lower bound of the network throughput increases as the spatial multiplexing gain $k$ increases. It illustrates that the overall performance of the mmWave network is improved by increasing $k$.
By comparing the upper bound and the tightest lower bound in Fig.\ref{fig:NT_vs_k}, it can be observed that the gap between the two bounds grows when the spatial multiplexing gain $k$ increases. Consequently, we can conclude that for the mmWave network, the performance of the fixed-rate coding scheme becomes less satisfactory as the spatial multiplexing gain $k$ increases.

\begin{figure}
	\vspace*{-0pt}
	\centering
	\subfloat[{\scriptsize Relative density $\psi = 4$ represents the dense deployments of mmWave APs.} \label{fig:NT_vs_k_psi4}]{%
		\includegraphics[width=0.48\linewidth]{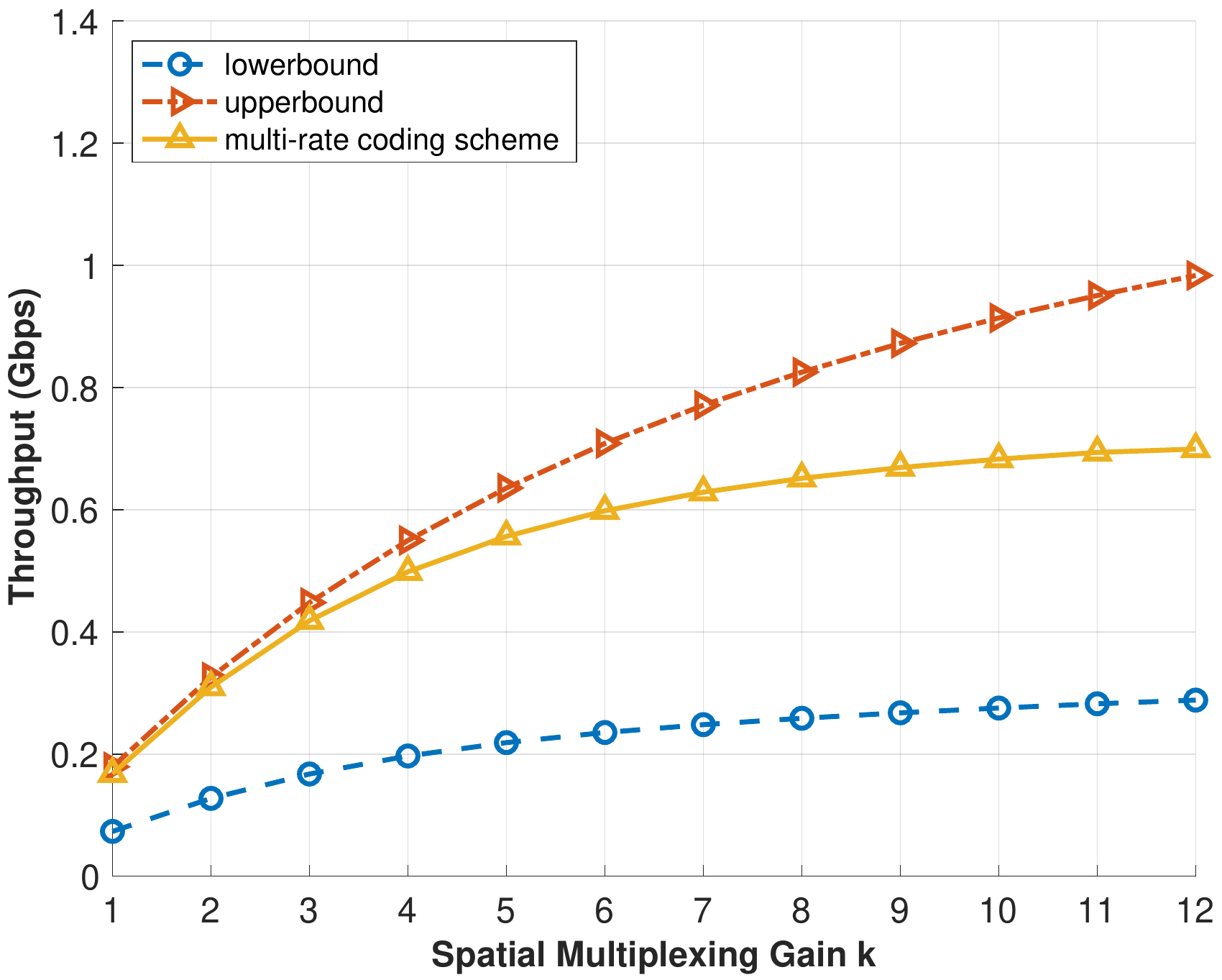}}
	\hfill
	\subfloat[{\scriptsize Relative density $\psi = 10$ represents the ultra-dense deployments of mmWave APs.}\label{fig:NT_vs_k_psi10}]{%
		\includegraphics[width=0.48\linewidth]{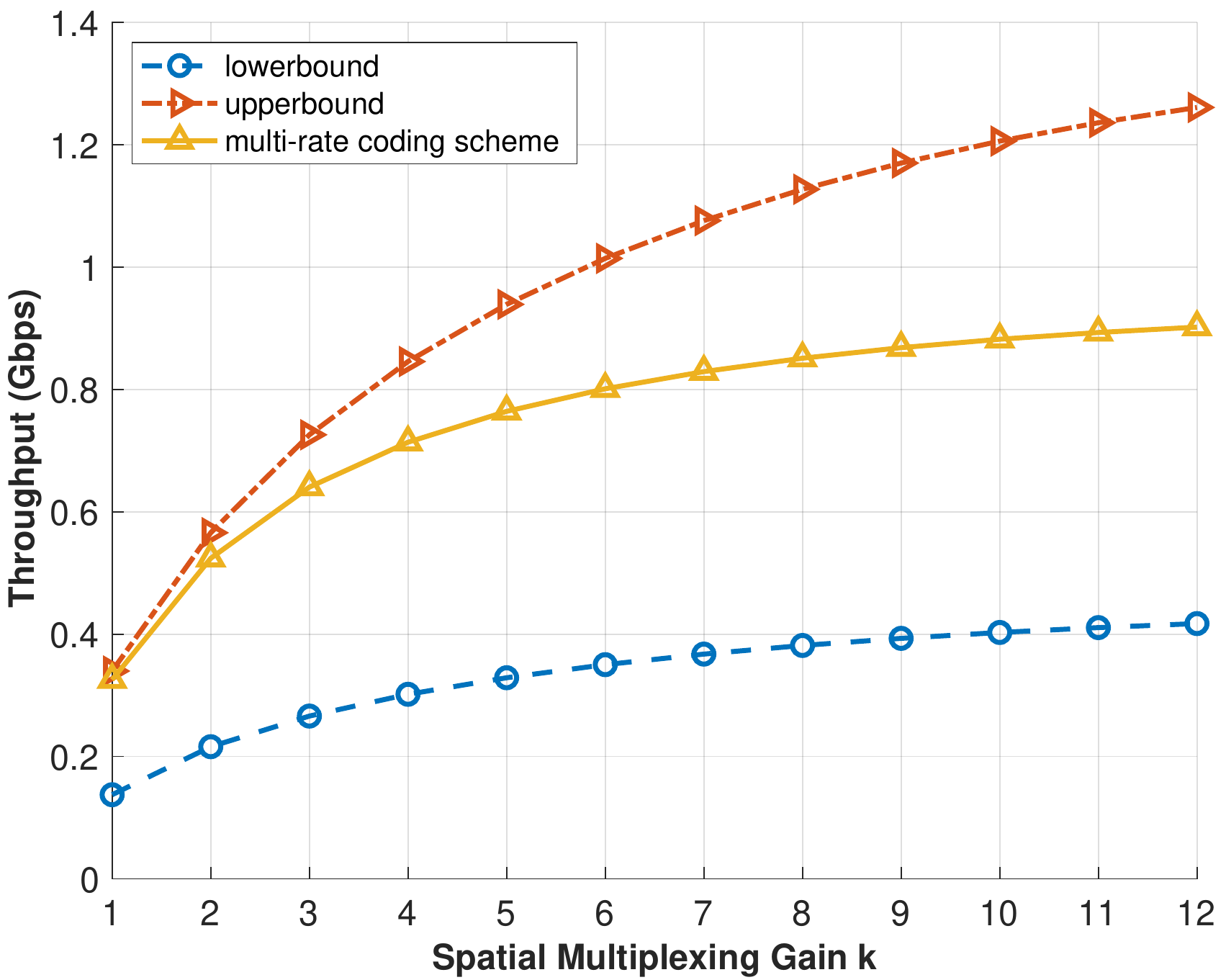}}
	\caption{\textcolor{black}{Comparison of the throughput upper and lower bounds for the mmWave network, where the throughput achieved by the multi-rate coding scheme lies in between.}}
	\label{fig:NT_vs_k} 
	\vspace*{-20pt}
\end{figure}

\subsection{Throughput of Multi-Rate Coding Scheme}
To achieve the throughput upper bound $\overline{\mathcal{T}}(\psi, k)$, each mmWave AP is required to deploy the coding scheme that can be instantaneously updated according to the SINR value of each end-user. However, the operational complexity of such a strategy is prohibitively high, which renders it impractical to be implemented in the mmWave network. 
In sharp contrast to the upper bound, the implementation of the fixed-rate coding scheme is simple. 
However, by fixing the rate threshold, any received rate higher than $\rho$ cannot be explored. Moreover, the received rate smaller than $\rho$ is always considered as an outage for the mmWave end-user. Therefore, the fixed-rate coding scheme results in a sub-optimal network throughput $\mathcal{T}(\rho; \psi, k)$.

Note that the multi-rate coding scheme defined in (\ref{eq: data rate of end-user - multi-rate coding scheme}) is a compromise between the operational complexity and the throughput performance \cite{alammouri2018unified}. It follows from (\ref{eq: throughput of multi-rate coding scheme}) that the multi-rate coding scheme is expected to achieve a throughput in between the upper and lower bounds. While the comprehensive study of a multi-rate coding scheme is out of the scope, in Fig.\ref{fig:NT_vs_k} we numerically show the throughput corresponding to the rate thresholds $\{\rho_i\}_{i=1}^{10}$ with  $\rho_i = 1+30(i-1)$ in units of MHz.
As can be seen, by employing such a multi-rate coding scheme, the mmWave network can achieve almost the throughput upper bound when the spatial multiplexing gain is small. Moreover, the throughputs provided by the multi-rate coding scheme increase by more than two times compared to the tightest lower bounds. Hence, we can conclude that the multi-rate coding scheme is necessary to guarantee the performance of mmWave network, especially when the spatial multiplexing gain at APs is large.

\section{Conclusion}
We proposed to model the downlink mmWave network incorporating the hybrid precoding gain, where both the analog beamforming gain and the spatial multiplexing gain $k$ were considered.
We then derived the coverage probability, the throughput of the fixed-rate coding scheme, and the throughput upper bound. With these analytical expressions, we separated the positive and negative effects of AP densification on mmWave networks.
Our analysis then proved that without the spatial multiplexing, over-densifying APs leads to the densification plateau for the network throughput. To overcome such a performance bottleneck, the mmWave APs need to deploy the spatial multiplexing, primarily because the spatial multiplexing gain improves the throughput by lowering the SINR threshold for mmWave end-users. 

We optimized the fixed-rate coding scheme to provide the throughput lower bound, while the throughput achieved by the multi-rate coding scheme was also evaluated for different $k$. For comparison, we also quantified the throughput upper bound. Our numerical results verified that all these throughputs increase as the spatial multiplexing gain $k$ increases. However, the gap between the throughput upper and lower bounds increases rapidly with $k$, which implied that the fixed-rate coding scheme is inefficient in exploiting the potential throughput gain for mmWave networks with large $k$. Further, the multi-rate coding scheme showed a satisfactory performance when $k$ grows, which triggered a potential study on the complexity-performance trade-off for the coding-scheme of mmWave networks.



%

\appendices
\section{\label{proof: coverage probability of LOS ball}}
With the LOS-ball model in (\ref{LOS-ball}), the coverage probability in (\ref{eq: LOS coverage probability}) can be written as (\ref{eq: simplified coverage probability}) with
\begin{eqnarray}
\Lambda_j(\tau,r,G) 
&\triangleq& 
\frac{2r}{\alpha_{\text{L}}} \left( j \eta \tau G \right)^{\frac{2}{\alpha_{\text{L}}}} 
\Gamma\left( -\frac{2}{ \alpha_{\text{L}}}; j \eta \tau G , j \eta \tau G r^{\frac{\alpha_{\text{L}}}{2}}\right)
\label{eq: SINR parameter definition}\\
&=&	
\frac{2r}{\alpha_{\text{L}}} \left( j \eta \tau G \right)^{\frac{2}{\alpha_{\text{L}}}} 
\int
_{j \eta \tau G r^{\frac{\alpha_{\text{L}}}{2}}}
^{j \eta \tau G}
t^{-\frac{2}{\alpha_{\text{L}}} - 1} e^{-t} 
\;\text{d}t.
\label{eq: taylor parameter}
\end{eqnarray}
Note that $ \Lambda_j(\tau,r,G)$ is in the form of incomplete gamma function.  

Here, we show the sketch of deriving the coverage probability of mmWave network with LOS-ball of radius $R_\text{B}$.
We start to derive the PDF of the distance from a end-user to its nearest AP in $\Phi_{\text{A}}$, which is denoted by $r$. Without loss of generality, consider the end-user at the origin. Following the network model, the APs which are LOS to the user form an inhomogeneous PPP with density $\lambda_{\text{A}}$ for $r\leq R_\text{B}$.  
It follows from \cite{andrews2011tractable} that the PDF and CDF of $r$ can be derived as
\begin{eqnarray}
f_{\text{L}}(r)  = 2 \pi r \lambda_\text{A} e^{-2 \pi \lambda_\text{A} \int_{0}^{r}xdx}, \nonumber\\
F_{\text{L}}(r)  = 1-e^{-2\pi \lambda_{\text{A}} \int_{0}^{r} {x \text{d}x} }.
\label{eq: f_l and F_l}
\end{eqnarray}
for $r\leq R_\text{B}$.
For the end-user at the origin, the coverage probability can be written as,
\begin{equation}
\mathcal{C}_\psi(\tau,k) \triangleq \mathbb{P}\left( \text{SINR}>\tau \left| \mathbb{P}(\mathbb{G}_\text{A} = G_\text{A}) = k\cdot\frac{\theta_{\text{A}}}{2\pi} \right. \right) = \mathbb{E}_{r}\left[\mathcal{C}_\psi (\tau,k,r)\right], \nonumber
\end{equation}
where $\mathcal{C}_\psi (\tau,k,r) \triangleq \mathbb{P}( \text{SINR}>\tau \;|\; |x_0| = r, \mathbb{P}(\mathbb{G}_\text{A} = G_\text{A}) = k\cdot\frac{\theta_\text{A}}{2\pi}   )$ is the coverage probability conditioned on the distance $r$ from the origin to the AP at $x_0$. 
To derive $\mathcal{C}_\psi (\tau,k,r)$, we need to first derive the Laplace transform of the interference $\mathcal{I}$ in (\ref{eq: interference of LOS}) given that $r = |x_0|$, as shown in \cite[]{andrews2011tractable}.
It follows from \cite[Theorem 4.9]{haenggi2012stochastic} that the Laplace transform of $\mathcal{I}$ evaluated at $s$, denoted by $\mathcal{L}_{\mathcal{I}}(s)$, can be given as
\begin{eqnarray}
\mathcal{L}_{\mathcal{I}}(s) \triangleq \mathbb{E} \left[ e^{-s \mathcal{I}} \right]
= \int_{r}^{R_\text{B}} \left(
1-\mathbb{E}
\left[
e^{- s h \mathbb{G}_\text{A} \mathbb{G}_\text{U}  x^{-\alpha_{\text{L}}}} 
\right]
\right)  x \; \text{d}x ,
\label{eq: Laplace Transform derivation}
\end{eqnarray}
where $h = 1$ means no fading in deterministic channel. As shown in \cite[Appendix D]{bai2015coverage}, $h$ can be considered as a gamma random variable with the shape parameter goes to infinity. 
By using Alzer's Lemma in \cite{alzer1997some}, the conditional coverage probability can be written as follows:
\begin{eqnarray}
\mathcal{C}_{\psi} (\tau,k,r)
&\approx&  
{\sum_{j=1}^{\mu}(-1)^{j+1}{\mu\choose j}}
\mathcal{L}_{\mathcal{I}}
\left( \tau \Omega j \right) ,
\label{eq: conditional SINR Nakagami LOS}
\end{eqnarray}
where $\mu$ is the shape parameter of channel fading $h$ and $\mu \geq 10$ is a good estimation on the deterministic channel \cite{bai2015coverage}; and $\Omega = \frac{\eta r^{\alpha_\text{L}} }{G_{\text{A}}G_{\text{U}}}$ with $\eta = \mu (\mu !)^{-\frac{1}{\mu}}$.
The rest of proof follows the same steps in \cite[Theorem 3]{bai2015coverage}.

\section{\label{proof: taylor expansion of coverage probability}}
The coefficients $c_0(\tau)$ and $c_l(\tau,\psi)$ in (\ref{eq: SINR interms of polynomial}) are defined as
\begin{eqnarray}
c_0 (\tau,\psi) = \psi e^{-\psi}
{\sum_{j=1}^{\mu}(-1)^{j+1}{\mu\choose j}}
\int_{0}^{1} 
\exp
\left\lbrace
\psi
\frac{\theta_{\text{U}}}{2\pi}
\Lambda_j\left( \tau,r,\frac{g_\text{A}}{G_\text{A}} \right) 
\right.\nonumber\\
+
\left.
\psi
\left(
1-\frac{\theta_{\text{U}}}{2\pi}
\right)
\Lambda_j\left(\tau,r,\frac{g_\text{A}g_\text{U}}{G_\text{A}G_\text{U}}\right)
\right\rbrace
\text{d} r,
\label{eq: coefficient c_0}
\end{eqnarray}
\begin{eqnarray}
c_l (\tau,\psi) &=& \frac{\psi e^{-\psi}}{l!}
{\sum_{j=1}^{\mu}(-1)^{j+1}{\mu\choose j}}
\int_{0}^{1}
e^
{
	\frac{\psi\theta_{\text{U}}}{2\pi}
	\Lambda_j\left( \tau,r,\frac{g_\text{A}}{G_\text{A}} \right)
	+
	\psi\left(1-	\frac{\theta_\text{U}}{2\pi} \right)
	\Lambda_j\left(\tau,r,\frac{g_\text{A}g_\text{U}}{G_\text{A}G_\text{U}}\right)
}
\nonumber\\
&&
\times
\left(
\frac{\psi\theta_{\text{A}}}{2\pi} 
\right)^l
\left[
\frac{\theta_{\text{U}}}{2\pi}
\Lambda_j\left( \tau,r,1\right)
\right.
\left.
+
\left(
1-\frac{\theta_{\text{U}}}{2\pi}
\right)
\Lambda_j\left(\tau,r,\frac{g_\text{U}}{G_\text{U}}\right)
\right.
\nonumber\\
&&
\left.
-
\frac{\theta_{\text{U}}}{2\pi}
\Lambda_j\left( \tau,r, \frac{g_\text{A}}{G_\text{A}}\right)
-
\left(
1-\frac{\theta_{\text{U}}}{2\pi}
\right)
\Lambda_j\left(\tau,r,\frac{g_\text{A}g_\text{U}}{G_\text{A}G_\text{U}}\right)
\right]^l
\text{d} r, \;\;\; l\in\mathbb{N}^+.
\label{eq: coefficient c_l}
\end{eqnarray}
Here, $\Lambda_j(\tau,r,G)$ is given in (\ref{eq: taylor parameter}). 

We start to prove that $c_0 = \mathbb{P}(\text{SNR} > \tau)$. By ignoring the thermal noise \cite{bai2015coverage}, the background noise is solely resulted by the side lobes of interfering APs. It is equivalent to letting the hybrid precoding gain $\mathbb{G}_\text{A} = g_\text{A}$ for all $x\in \Phi_{\text{A}}\backslash \{x_0\}$. By plugging $\mathbb{G}_\text{A} = g_\text{A}$ into the coverage probability in (\ref{eq: SINR interms of polynomial}), the expression of $c_0(\tau,\psi)$ immediately follows.

Next, we derive the Taylor series expansion for the expression in (\ref{eq: simplified coverage probability}). 	 
By denoting $
\mathcal{F}_{\psi}(k,\tau,j)
=
\int_{0}^{1} \exp 
\left\lbrace
\psi 
\mathbb{E}
\left[
\Lambda_j
\left(\tau,r,\frac{\mathbb{G}_\text{A}\mathbb{G}_\text{U}}{G_\text{A}G_\text{U}} \right)
\right]
\right\rbrace
\text{d} r$,
the coverage probability in (\ref{eq: simplified coverage probability}) can then be written as
\begin{equation}
\mathcal{C}_{\psi}(\tau,k) 
=
\psi e^{-\psi}{\sum_{j =1}^{\mu}(-1)^{j+1}{\mu\choose j}} \mathcal{F}_{\psi}(k,\tau,j).
\label{proofeq: coverage probability with F}
\end{equation}
By plugging in the hybrid precoding gain $\mathbb{G}_\text{A}$ and $\mathbb{G}_\text{U}$, the term $\mathcal{F}_{\psi}(k,\tau,j)$ can be further expanded as follows:
\begin{eqnarray}
\mathcal{F}_{\psi}(k,\tau,j)
&=&
\int_{0}^{1}
\exp
\left\lbrace
\psi
\left[
k
\frac{\theta_{\text{A}}\theta_{\text{U}}}{4\pi^2}
\Lambda_j\left( \tau,r,1\right)
+
k
\frac{\theta_{\text{A}}}{2\pi}
\left(
1-\frac{\theta_{\text{U}}}{2\pi}
\right)
\Lambda_j\left(\tau,r,\frac{g_\text{U}}{G_\text{U}}\right)
\right.
\right.
\nonumber\\
&&
\left.
\left.
+
\left(
1-\frac{k\theta_{\text{A}}}{2\pi}
\right)
\frac{\theta_{\text{U}}}{2\pi}
\Lambda_j\left( \tau,r, \frac{g_\text{A}}{G_\text{A}}\right)
+
\left(
1-\frac{k\theta_{\text{A}}}{2\pi}
\right)
\left(
1-\frac{\theta_{\text{U}}}{2\pi}
\right)
\Lambda_j\left(\tau,r,\frac{g_\text{A}g_\text{U}}{G_\text{A}G_\text{U}}\right)
\right]
\right\rbrace
\text{d} r.
\nonumber
\end{eqnarray}
It follows that the $l^\text{th}$ derivative of $\mathcal{F}_{\psi}(k,\tau,j)$ can be written as
\begin{eqnarray}
\mathcal{F}_{\psi}^{(l)}(k,\tau,j)
&=&
\int_{0}^{1}
\frac{\partial^l}{\partial k^l}
\exp
\left\lbrace
\psi
\left[
k
\frac{\theta_{\text{A}}\theta_{\text{U}}}{4\pi^2}
\Lambda_j\left( \tau,r,1\right)
+
k
\frac{\theta_{\text{A}}}{2\pi}
\left(
1-\frac{\theta_{\text{U}}}{2\pi}
\right)
\Lambda_j\left(\tau,r,\frac{g_\text{U}}{G_\text{U}}\right)
\right.
\right.
\nonumber\\
&&
\left.
\left.
+
\left(
1-\frac{k\theta_{\text{A}}}{2\pi}
\right)
\frac{\theta_{\text{U}}}{2\pi}
\Lambda_j\left( \tau,r, \frac{g_\text{A}}{G_\text{A}}\right)
+
\left(
1-\frac{k\theta_{\text{A}}}{2\pi}
\right)
\left(
1-\frac{\theta_{\text{U}}}{2\pi}
\right)
\Lambda_j\left(\tau,r,\frac{g_\text{A}g_\text{U}}{G_\text{A}G_\text{U}}\right)
\right]
\right\rbrace
\text{d} r,
\nonumber
\end{eqnarray}
which follows from the Leibniz integral rule.
Thus, the Taylor series expansion of $\mathcal{F}_{\psi}(k,\tau, j)$ at $k = 0$ can be written as
\begin{eqnarray}
&&\mathcal{F}_{\psi}(k,\tau, j) = \sum\limits_{l=0}^{\infty} 
\frac{\mathcal{F}_{\psi}^{(l)}(0,\tau,j)}{l!} k^l \nonumber\\
&=& 
\sum_{l=0}^{\infty}
\frac{k^l}{l!}
\int_{0}^{1}
e^
{
	\frac{\psi\theta_{\text{U}}}{2\pi}
	\Lambda_j\left( \tau,r,\frac{g_\text{A}}{G_\text{A}} \right)
	+
	\left(\psi-	\frac{	\psi \theta_\text{U}}{2\pi} \right)
	\Lambda_j\left(\tau,r,\frac{g_\text{A}g_\text{U}}{G_\text{A}G_\text{U}}\right)
}
\left(
\frac{\psi\theta_{\text{A}}}{2\pi} 
\right)^l
\left[
\frac{\theta_{\text{U}}}{2\pi}
\Lambda_j\left( \tau,r,1\right)
+
\left(
1-\frac{\theta_{\text{U}}}{2\pi}
\right)
\Lambda_j\left(\tau,r,\frac{g_\text{U}}{G_\text{U}}\right)
\right.
\nonumber\\
&&
\left.
-
\frac{\theta_{\text{U}}}{2\pi}
\Lambda_j\left( \tau,r, \frac{g_\text{A}}{G_\text{A}}\right)
-
\left(
1-\frac{\theta_{\text{U}}}{2\pi}
\right)
\Lambda_j\left(\tau,r,\frac{g_\text{A}g_\text{U}}{G_\text{A}G_\text{U}}\right)
\right]^l
\text{d} r.
\label{proofeq: Taylor expansion}
\end{eqnarray}
Now we prove that the Taylor series in (\ref{proofeq: Taylor expansion}) converge to  $\mathcal{F}_{\psi}(k,\tau, j)$ in the region $k \in [0,+\infty)$. Following the definition of $\Lambda_j\left( \tau,r,G\right)$ in (\ref{eq: taylor parameter}), we have
\begin{eqnarray}
0 < \Lambda_j\left( \tau,r,G\right)
\stackrel{(a)}{<}\frac{2r}{\alpha_{\text{L}}} \left( j \eta \tau G \right)^{\frac{2}{\alpha_{\text{L}}}} 
\int
_{j \eta \tau G r^{\frac{\alpha_{\text{L}}}{2}}}
^{j \eta \tau G}
t^{-\frac{2}{\alpha_{\text{L}}} - 1} 
\;\text{d}t = 1-r, \quad\forall G>0,\; r\in(0,1],
\label{proofeq: Taylor parameter upper bound}
\end{eqnarray}
where (a) follows from $ e^{-t}< 1$ for $t > 0$.
Then, the absolute value of $\mathcal{F}_{\psi}^{(l)}(0,\tau,j)$ can be upper bounded by 
\begin{eqnarray}
|\mathcal{F}_{\psi}^{(l)}(0,\tau,j)|
< 
\left(
\frac{\psi\theta_{\text{A}}}{2\pi} 
\right)^l
\int_{0}^{1}
e^
{ \psi(1-r) }
(1-r)^l
\text{d} r
< \frac{e^\psi}{l+1}\left(
\frac{\psi\theta_{\text{A}}}{2\pi} 
\right)^l.
\label{proofeq: Taylor expansion upper bound}
\end{eqnarray}
Therefore, $\mathcal{F}_{\psi}(k,\tau, j)$ is real analytic on $k \in [0,+\infty)$. 
Define 
\begin{equation}
c_l (\tau, \psi) \triangleq \frac{\psi e^{-\psi}}{l!}
{\sum_{j=1}^{\mu}(-1)^{j+1}{\mu\choose j}}
\mathcal{F}_{\psi}^{(l)}(0,\tau,j).
\label{proofeq: c_l in F}
\end{equation}
The expression of $c_l(\tau,\psi)$ for $l \in \mathbb{N}^+$ in (\ref{eq: coefficient c_l}) then follows.

Define 
$c_l (\tau, \psi, k) \triangleq \frac{\psi e^{-\psi}}{l!}
{\sum_{j=1}^{\mu}(-1)^{j+1}{\mu\choose j}}
\mathcal{F}_{\psi}^{(l)}(k,\tau,j)$, where $k \in [1,K]$ and $K= \frac{2\pi}{\theta_{\text{A}}}$.
Next, we start to derive the bound for $\mathcal{F}_{\psi}^{(l)}(k,\tau,j)$ i.e. the $l^\text{th}$ derivative of $\mathcal{F}_{\psi}(k,\tau,j)$. By following the similar steps in the derivation of (\ref{proofeq: Taylor expansion upper bound}), we have
 \begin{eqnarray}
 |\mathcal{F}_{\psi}^{(l)}(k,\tau,j)|
 < \frac{e^{2\psi}}{l+1}\left(
 \frac{\psi\theta_{\text{A}}}{2\pi} 
 \right)^l.
 \label{proofeq: Parameter F bound}
 \end{eqnarray}
Given the degree of polynomial $L$, the approximation error is upper bounded by
\begin{eqnarray}
& \max\limits_{k\in[1,K],\tau\in\mathbb{R}^+} &\left| \mathcal{C}_{\psi}(\tau,k) -  \sum\limits_{l=0}^{L} c_l (\tau,\psi) k^l \right| 
 =  \left|  \sum\limits_{l=L+1}^{\infty} c_l (\tau,\psi) K^l \right| 
 \nonumber\\
&\stackrel{(a)}{=}&
K^{L+1}
| c_{L+1}(\tau, \psi, k^*) |
\qquad
\text{for some}\; k^* \in [1,K] \nonumber \\
&\stackrel{(b)}{<}&
\frac{ e^\psi \psi^{L+2} }{(L+2)!}
{\sum_{j=1}^{ \mu}
{\mu\choose j }}
\left(
\frac{K \theta_{\text{A}}}{2\pi} 
\right)^{L+1} \nonumber
\stackrel{(c)}{=}
\frac{(2^{\mu}-1)  e^\psi \psi^{L+2}}{(L+2)!},
\nonumber
\end{eqnarray}
where (a) follows the Lagrange’s formula; (b) follows the inequality (\ref{proofeq: Parameter F bound}); (c) follows $K = \frac{2\pi}{\theta_{\text{A}}}$.

\section{\label{proof: bandwidth distribution}}
Consider the mmWave network, where the channel follows the path loss model $r^{-\alpha_{\text{L}}}$ with no fading. According to the minimum path loss association rule, the serving area of each AP can be modeled by a Voronoi cell \cite{andrews2011tractable}. For $\Phi_{\text{A}}$ of intensity $\lambda_{\text{A}}$, the size distribution of a Voronoi cell is derived in \cite{ferenc2007size}, 
with the PDF 
\begin{equation}
f(y) \propto y^{2.5}\exp{(-3.5{\lambda_{\text{A}}}y)}, \;\;\forall\; y>0.
\label{proofeq: Voronoi cell} 
\end{equation}
Based on the PDF of Voronoi cell in (\ref{proofeq: Voronoi cell}), we then derive the size distribution of an AP sector. Let a randomly selected AP be the origin. It follows from the isotropic property of the PPP $\Phi_\text{U}$ that the end-users are isotropically distributed around the origin. Thus, the serving region of the selected AP is isotropically distributed around the origin. The area distribution of an AP sector can then be written as
\begin{equation}
f_{\mathcal{A}}(y) =
\left\lbrace 
\begin{split}
&\frac{3.5^{3.5}{k\lambda_{\text{A}}}}{\Gamma(3.5, \frac{\psi}{3.5k^2\lambda_{\text{A}}^2})}
({k\lambda_{\text{A}}}y)^{2.5}\exp{(-3.5{k\lambda_{\text{A}}}y)}, 
\; \text{for}\; y \leq \frac{\psi}{k\lambda_{\text{A}}}\\
&0, \; \text{for} \; y > \frac{\psi}{k\lambda_{\text{A}}}
\end{split}.
\right.
\label{eq: the area distribution of AP sector}
\end{equation}
Here, $\Gamma(s,x) = \int_{0}^{x} t^{s-1}e^{-t} \text{d}t$ is the incomplete gamma function.
We remark that the serving area of a mmWave AP is restricted to the LOS-ball with radius $R_\text{B}$. Therefore, for the mmWave AP sector, $f_{\mathcal{A}}(y) = 0$ when $y>\frac{1}{k} \pi R_\text{B}^2 = \frac{\psi}{k\lambda_{\text{A}}}$.

The end-user located at the origin is served by the tagged AP. We use $\mathcal{A}'$ to denote the serving area within a sector of the tagged AP. It can be proved by a minor modification of \cite[lemma2]{yu2013downlink} that the distributions of $\mathcal{A}'$ and $\mathcal{A}$ in (\ref{eq: the area distribution of AP sector}) are related as follows:
\begin{eqnarray}
f_{\mathcal{A}'}(y)& \propto& y f_{\mathcal{A}}(y) \label{biased sector} \nonumber\\
&=& \frac{3.5^{4.5}{k\lambda_\text{A}}}{\Gamma(4.5, \frac{\psi}{3.5k^2\lambda_{\text{A}}^2})}({k\lambda_\text{A}}y)^{3.5}\exp{(-3.5{k\lambda_\text{A}}y)}, y \leq \frac{\psi}{\lambda_{\text{A}}},
\label{proofeq: biased sector 2}
\end{eqnarray}
where $\Gamma(s,x)$ is the incomplete gamma function $\Gamma(s,x) = \int_{0}^{x} t^{s-1}e^{-t} \text{d}t$. It follows from (\ref{proofeq: biased sector 2}) that the area distribution of the tagged AP sector is biased since the end-user at the origin is more likely to lie in the larger area. 

Note that $\Psi$ represents the number of end-users that share the same pool of radio resources as the origin. 
By following \cite{singh2013offloading}, the PMF of $\Psi$ can be obtained as follows: 
\begin{eqnarray}
\mathbb{P}({\Psi}= n) 
&=& \mathcal{K}_\text{T}(n,k; \psi,\lambda_{\text{U}}) 
=  \frac{\left. \left(\int_{0}^{\psi/\lambda_{\text{A}}} {\exp{(\lambda_\text{U}y(z-1))}} f_{\mathcal{A}'}(y) \, \text{d}y\right)^{(n-1)} \right|_{z=0}}{(n-1)!} \nonumber\\
&=&
\frac{3.5^{4.5} \; \Gamma\left(n+3.5, \frac{\psi}{k\lambda_{\text{A}}(\lambda_{\text{U}}+3.5k\lambda_\text{A})}\right)}
{(n-1)! \; \Gamma\left(4.5,\frac{\psi}{3.5k^2\lambda_{\text{A}}^2}\right)}
\left(\frac{\lambda_\text{U}}{k\lambda_\text{A}}\right)^{n-1}
\left(3.5+\frac{\lambda_\text{U}}{k\lambda_\text{A}}\right)^{-n-3.5}, n\geq 1,
\label{eq: tagged end-user load}
\end{eqnarray} 
where $\lambda_{\text{A}} = \psi/\pi R_\text{B}^2$ with $R_\text{B}$ being the radius of LOS ball; $\Gamma(s,x)$ is the incomplete gamma function $\Gamma(s,x) = \int_{0}^{x} t^{s-1}e^{-t} \text{d}t$.
Note that $n \geq 1$ since $\Psi$ always contains the end-user at the origin.
By assuming the relative density $\psi \geq 1$, the end-user can observe at least one LOS AP \cite{bai2015coverage}, which results in 
$\Gamma\left(n+3.5, \frac{\psi}{\lambda_{\text{A}}(\lambda_{\text{U}}+3.5k\lambda_\text{A})}\right) \approx \Gamma(n+ 3.5)$ and
$\Gamma\left(4.5,\frac{\psi}{3.5k\lambda_{\text{A}}^2}\right) \approx \Gamma(4.5)$. It follows that
\begin{equation}
\mathbb{P}\left(\Psi= n\right) =\frac{3.5^{3.5} \; \Gamma\left( n+3.5 \right)}
{(n-1)! \; \Gamma\left( 3.5 \right)}
\left(\frac{\pi R_\text{B}^2 \lambda_\text{U}}{k \psi}\right)^{n-1}
\left(3.5+\frac{\pi R_\text{B}^2 \lambda_\text{U}}{k \psi}\right)^{-n-3.5}, n\geq 1,
\label{eq: simplied tagged load}
\end{equation}
where $\Gamma(s) = \int_{0}^{+\infty} t^{s-1}e^{-t} \text{d}t$ is the gamma function; $R_\text{B}$ is the radius of the LOS-ball.

\section{\label{proof: throughput upper bound}}
To derive the analytical expression of the throughput upper bound, we notice that
\begin{eqnarray}
 \mathbb{E} [ \log_2 (1 + \text{SINR})]
&\stackrel{(a)}{=}& \log_2(e) \int_{0}^{\infty} {\mathbb{P}(\ln (1 + \text{SINR}) > t)} \; \text{d}t \nonumber \\
&\stackrel{(b)}{=}& \log_2(e) \int_{0}^{\infty} {\frac{\mathbb{P}(\text{SINR} > \tau)}{1 + \tau}} \; \text{d}\tau 
= \log_2(e) \int_{0}^{\infty} {\frac{\mathcal{C}_{\psi}(\tau,k)}{1 + \tau}} \; \text{d}\tau,
\label{proofeq: upper bound with P_cov}
\end{eqnarray}
where (a) follows from that for a positive random variable $X$, $\mathbb{E}(X) = \int_{t>0} \; \mathbb{P}(X > t) \; \text{d}t$; (b) is derived by the change of random variable $t = \ln(\tau + 1)$. It follows from Theorem \ref{thm: taylor expansion of coverage probability} that
\begin{eqnarray}
\mathbb{E} [ \log_2 (1 + \text{SINR})]= \log_2(e) 
\sum\limits_{l=0}^{\infty}  k^l  
\int_{0}^{\infty} \frac{c_l (\tau,\psi)}{1+\tau} \text{d}\,\tau , 
\label{eq: spectral efficiency upper bound} 
\end{eqnarray}
We then complete the proof of (\ref{eq: spectral efficiency upper bound}) by showing $\int_{0}^{\infty} \frac{c_l (\tau,\psi)}{1+\tau} \text{d}\,\tau < +\infty$. To that end, we first derive the bound for $\Lambda_j\left( \tau,r,G\right)$ in (\ref{eq: taylor parameter}), where
\begin{eqnarray}
0 < \Lambda_j\left( \tau,r,G\right)
< e^{-j\eta\tau G r^{\alpha_{\text{L}}/2}}(1-r), \quad\forall G>0,\; r\in(0,1].
\label{proofeq: tighter bound on Taylor parameter}
\end{eqnarray}
It follows from (\ref{proofeq: Taylor parameter upper bound}) that 
$\mathcal{F}_{\psi}^{(l)}(0,\tau,j)$ in Appendix \ref{proof: taylor expansion of coverage probability} can be upper bounded by
\begin{eqnarray}
|\mathcal{F}_{\psi}^{(l)}(0,\tau,j)|
&<& 
\left(
\frac{\psi\theta_{\text{A}}}{2\pi} 
\right)^l
\left[
\int_{0}^{\frac{1}{(1+\tau)^{1/\alpha_{\text{L}}}}}
e^
{ \psi(1-r) }
(1-r)^l
\text{d} r
+
\int_{\frac{1}{(1+\tau)^{1/\alpha_{\text{L}}}}}^{1}
e^
{ \psi(1-r)-j\eta G l \tau r^{\alpha_\text{L}/2} }
(1-r)^l
\text{d} r
\right] \nonumber\\
&<&
e^\psi 
\left(
\frac{\psi\theta_{\text{A}}}{2\pi} 
\right)^l
\left(
\frac{1}{(1+\tau)^{\frac{1}{\alpha_{\text{L}}}}} + e^{-j\eta G l\frac{\tau}{\sqrt{1+\tau}}}
\right),
\end{eqnarray}
which illustrates $\int_{0}^{\infty} \frac{|\mathcal{F}_{\psi}^{(l)}(0,\tau,j)|}{1+\tau} \text{d}\,\tau < +\infty$. Further, $c_l(\tau,\psi)$ in (\ref{proofeq: c_l in F}) is a finite sum of $\mathcal{F}_{\psi}^{(l)}(0,\tau,j)$ and thus $\int_{0}^{\infty} \frac{c_l (\tau,\psi)}{1+\tau} \text{d}\,\tau < +\infty$.
The throughput upper bound in (\ref{eq: throughput upper bound definition}) can then be given as
\begin{eqnarray}
\overline{\mathcal{T}}(\psi, k) 
= 
\lambda_\text{U} \mathbb{E}\left[\frac{B}{\Psi}\log_2(1 + \text{SINR})\right] 
&\stackrel{(a)}{=}& \lambda_\text{U} \mathbb{E}\left[\frac{B}{\Psi}\right] \mathbb{E}\left[ \log_2(1 + \text{SINR}) \right],
\label{eq: upper bound proof 2}
\end{eqnarray}
where (a) follows from that the SINR distribution is independent of the bandwidth distribution. 
By plugging (\ref{eq: spectral efficiency upper bound}) into (\ref{eq: upper bound proof 2}), the result immediately follows.

\section{\label{proof: lem - densification plateau}}
First, we need to derive the asymptotic SINR distribution of mmWave network as the relative density $\psi \rightarrow +\infty$. 
Following (\ref{eq: interference of LOS}) and (\ref{eq: LOS coverage probability}), the coverage probability at SINR threshold $\tau$ is lower and upper bounded by assuming all interfering hybrid precoding gain in (\ref{eq: interference of LOS}) as $\mathbb{G}_\text{A}\mathbb{G}_\text{U} = G_\text{A}G_\text{U}$ and $\mathbb{G}_\text{A}\mathbb{G}_\text{U} = g_\text{A}g_\text{U}$, respectively. It follows that the coverage probability in (\ref{eq: LOS coverage probability}) is bounded by
\begin{equation}
\mathbb{P} \left( 
\frac{|x_0|^{-\alpha_{\text{L}}} }{\sum_{x \in \Phi_{\text{A}} \backslash \mathcal{B}(0,|x_0|)}  |x|^{-\alpha_{\text{L}}}}> \tau \right) 
<
\mathcal{C}_\psi(\tau, k)
<
\mathbb{P} \left( 
\frac{|x_0|^{-\alpha_{\text{L}}} }{\sum_{x \in \Phi_{\text{A}} \backslash \mathcal{B}(0,|x_0|)}  |x|^{-\alpha_{\text{L}}}}> \hat{\tau} \right) ,
\label{proofeq: bounds of SIR in the asymptotic region}
\end{equation}
where $x_0$ is the location of the AP serving the origin; $\hat{\tau} = \tau \frac{ g_\text{A} g_\text{U} }{ G_\text{A} G_\text{U} }$; the thermal noise is ignored since the signal from the side lobe of interfering APs is the major component of background noise when $\psi$ is large. 

Note that we assume that $R_\text{B}$ is fixed for the mmWave network and thus the relative density $\psi$ increases with the AP density $\lambda_{\text{A}}$. For a given $x_0 \in \mathbb{R}^2$ and assume $\ell(|x_0|) = 1$ for $|x_0| < 1$, 
\begin{eqnarray}
\lim\limits_{\lambda_{\text{A}} \rightarrow +\infty} \lambda_{\text{A}} \text{SINR} 
&=& 
\lim\limits_{\lambda_{\text{A}} \rightarrow +\infty} \lambda_{\text{A}} \frac{|x_0|^{-\alpha_{\text{L}}} }{\sum_{x \in \Phi_{\text{A}} \backslash \mathcal{B}(0,|x_0|)}  |x|^{-\alpha_{\text{L}}}} \nonumber\\
&\stackrel{(a)}{=}& \lim\limits_{\lambda_{\text{A}} \rightarrow +\infty} \frac{\lambda_{\text{A}} |x_0|^{-\alpha_\text{L}} }{2 \pi \lambda_{\text{A}} \int_{0}^{R_\text{B}} r^{-\alpha_\text{L}+ 1}  \text{d}r} 
\stackrel{(b)}{=}
\frac{ 1}{2 \pi  \sigma },
\label{proofeq: limitation of SIR - upper}
\end{eqnarray}
where $\sigma = \int_{0}^{R_\text{B}} r^{-\alpha_\text{L}+ 1} \text{d}r$ is a finite number; (a) follows from \cite[Lemma 1]{alammouri2018unified}; (b) follows that $\ell(|x_0|) \geq 1$.
When the spatial multiplexing gain $k=1$ and $\psi$ increases, we then have the densification gain $\gamma(\psi,1)$ in (\ref{def: densification gain}) as
\begin{eqnarray}
	\lim\limits_{\psi \rightarrow +\infty} \gamma(\psi,1) 
	&=& \frac{1}{\mathcal{T}(\rho_0;1,1)} \lim\limits_{\psi \rightarrow +\infty}   \mathcal{T}(\psi\rho_0; \psi, 1) \nonumber \\
	&\stackrel{(a)}{=}& \frac{\lambda_\text{U}\rho_0}{\mathcal{T}(\rho_0;1,1)} \lim\limits_{\psi \rightarrow +\infty}   \psi\mathcal{C}_\psi(2^{\rho_0/W_0}-1,1) \nonumber\\
	&\stackrel{(b)}{\leq}& \frac{\lambda_\text{U}\rho_0}{\pi R_\text{B}^2 \mathcal{T}(\rho_0;1,1)} \lim\limits_{\lambda_\text{A} \rightarrow +\infty}   \lambda_{\text{A}} \mathbb{P}\left(\text{SINR} > \hat{\tau}\right) \nonumber\\
	&=&  \frac{\lambda_\text{U}\rho_0}{\pi R_\text{B}^2 \mathcal{T}(\rho_0;1,1)} \lim\limits_{\lambda_\text{A} \rightarrow +\infty}   \lambda_{\text{A}} \mathbb{E}\left( \mathbbm{1} (\text{SINR} > \hat{\tau}) \right) \nonumber \\
	& = & \frac{\lambda_\text{U}\rho_0}{\pi R_\text{B}^2 \mathcal{T}(\rho_0;1,1)} \mathbb{E}\left( \lim\limits_{\lambda_\text{A} \rightarrow +\infty}\lambda_{\text{A}}\mathbbm{1} \left(\frac{\text{SINR}}{\hat{\tau}} >1 \right) \right) \nonumber\\
	&\leq & \frac{\lambda_\text{U}\rho_0}{\pi R_\text{B}^2 \mathcal{T}(\rho_0;1,1)} \mathbb{E}\left( 
	\lim\limits_{\lambda_\text{A} \rightarrow +\infty}\lambda_{\text{A}}\text{SINR} / \hat{\tau} \right) \nonumber\\
	&\stackrel{(c)}{\leq}& \frac{\lambda_\text{U}\rho_0}{2\pi^2 R_\text{B}^2 \sigma\hat{\tau} \mathcal{T}(\rho_0;1,1)},
	\label{proofeq: limitation of densification gain with k=1}
\end{eqnarray}
where $\hat{\tau} = (2^{\rho_0/W_0}-1) \frac{ g_\text{A} g_\text{U} }{ G_\text{A} G_\text{U} }$; (a) is obtained by plugging (\ref{eq: average bandwidth}) into (\ref{eq: throughput of fixed-rate coding scheme}); (b) follows from (\ref{proofeq: bounds of SIR in the asymptotic region}); (c) follows from (\ref{proofeq: limitation of SIR - upper}). It follows from (\ref{proofeq: limitation of densification gain with k=1}) that $\gamma(\psi,1)$ is upper bounded by a finite constant.



\ifCLASSOPTIONcaptionsoff
  \newpage
\fi

\small
\bibliographystyle{IEEEtran}
\bibliography{mmWavepapers}
\end{document}